\documentclass[version=preprint, biblatex]{iacrcc}
\usepackage{amsmath, amsthm, amsfonts, relsize, units, amssymb, setspace, enumitem, float, booktabs, tabularx, color, boxedminipage, latexsym, mathtools, enumitem, float, parskip, blindtext, newclude, algorithm, algorithmicx, algpseudocode, graphicx}
\title[running  = {Withdrawable FS},
       subtitle = {}
      ]{Withdrawability in Fiat–Shamir with aborts constructions}

\addauthor[orcid   = {0000-0002-8959-636X},
           inst    = {1},
           email   = {ramses.fernandez@fairgate.io},
           surname = {Fernandez},
          ]{Ramses Fernandez}

\addaffiliation[country = {}]{Fairgate Labs}

\authorrunning{Ramses Fernandez}
\keywords[Withdrawability, Digital signatures, Module Learning with Errors]{Withdrawability, Digital signatures, Module Learning with Errors}

\newcommand{\Ring}{R}
\newcommand{\Rq}{R_q}
\newcommand{\bA}{\mathbf{A}}
\newcommand{\bB}{\mathbf{B}}
\newcommand{\bs}{\mathbf{s}}
\newcommand{\by}{\mathbf{y}}
\newcommand{\bz}{\mathbf{z}}
\newcommand{\bw}{\mathbf{w}}
\newcommand{\bt}{\mathbf{t}}
\newcommand{\bI}{\mathbf{I}}
\newcommand{\sample}{\xleftarrow{\$}}
\newcommand{\modpm}{\mathbin{\mathrm{mod}^{\pm}}}
\newcommand{\HighBits}{\mathsf{HighBits}}
\newcommand{\LowBits}{\mathsf{LowBits}}
\newcommand{\KeyGen}{\mathsf{KeyGen}}
\newcommand{\Sign}{\mathsf{Sign}}
\newcommand{\Verify}{\mathsf{Verify}}
\newcommand{\WSign}{\mathsf{WSign}}
\newcommand{\WSVerify}{\mathsf{WSVerify}}
\newcommand{\Confirm}{\mathsf{Confirm}}
\newcommand{\CVerify}{\mathsf{CVerify}}
\newcommand{\Simul}{\mathsf{Simul}}
\newcommand{\Adv}{\mathsf{Adv}}
\newcommand{\Exp}{\mathsf{Exp}}
\newcommand{\Ora}[1]{\mathcal{O}^{\mathsf{#1}}}
\newcommand{\negl}{\mathsf{negl}}
\newcommand{\norm}[1]{\lVert #1 \rVert}
\newcommand{\norminf}[1]{\lVert #1 \rVert_{\infty}}
\newcommand{\adv}{\mathcal{A}}

\newcommand{\PKE}{\mathsf{PKE}}
\newcommand{\Dec}{\mathsf{Dec}}
\newcommand{\Enc}{\mathsf{Enc}}
\newcommand{\eufcma}{\mathrm{EUF\text{-}CMA}}
\newcommand{\eufic}{\mathrm{EUF\text{-}CMA\text{-}ic}}
\newcommand{\Setup}{\mathsf{Setup}}
\DeclareMathOperator{\SVP}{SVP}

\newtheorem{assumption}{Assumption}

\begin{document}
\maketitle

\begin{abstract}
This article presents an extension of the work performed by Liu, Baek and Susilo \cite{liu2023} on withdrawable signatures to the Fiat--Shamir with aborts paradigm. We introduce an abstract construction, and provide security proofs for this proposal. As an instantiation, we provide a concrete withdrawable signature scheme based on a no-hint, full-$\bt$ Dilithium-style Fiat--Shamir-with-aborts construction \cite{ducas2017}; adapting to production ML-DSA (with hints) introduces a small $\varepsilon_{\mathrm{zk}}$ term.
\end{abstract}

\section{Introduction}
Digital signatures serve as a fundamental cryptographic mechanism that enables entities to bind their identities to pieces of information. The essential purpose of a digital signature is to allow a signer, who has established a public key $pk$, to sign messages using their private key $sk$ in a way that enables anyone knowing $pk$ to verify both the message's origin and its integrity during transit.

An important paradigm for the creation of digital signatures is the Fiat--Shamir transform, which converts interactive identification protocols into non-interactive digital signature schemes. Starting with a three-move identification protocol, where a prover demonstrates knowledge to a verifier through commitment, challenge, and response steps, the transform replaces the verifier's random challenge with a hash function applied to both the commitment and the message. This creates a digital signature scheme where the signing algorithm computes a commitment, generates a challenge by hashing the commitment with the message, and produces a response using the secret key.

The Schnorr signature scheme is perhaps the most well-known application of the Fiat--Shamir transform, which has gained particular attention due to its security characteristics and its valuable properties, such as signature aggregation. These advantages make Schnorr signatures especially attractive for blockchain applications where transaction size reduction and privacy enhancement are crucial considerations.

The impact of digital signatures is particularly important in blockchain technology, where this primitive extends beyond basic transaction authentication, enabling sophisticated smart contract interactions, multi-signature schemes for enhanced security, and threshold signature systems for distributed key management. Furthermore, innovations in signature aggregation and batch verification techniques have contributed significantly to blockchain scalability solutions.

Current public-key cryptographic algorithms serve as the foundation for protecting sensitive electronic information from unauthorized access. These algorithms have successfully withstood attacks from conventional computing systems for decades, due to the hardness of their underlying mathematical problems, prime factorization and the computation of discrete logarithms. However, the emergence of quantum computing presents a significant challenge to this security paradigm, as Shor's algorithm demonstrates the potential to solve both the prime factorization and the computation of discrete logarithms efficiently. This means that quantum computers possess computational capabilities that could potentially compromise current cryptographic methods, exposing vulnerable data and information.

To address this impending challenge, new cryptographic approaches are being designed to withstand attacks from both traditional computers and future quantum systems. These methods rely on problems, such as lattices, error-correcting codes or isogenies of elliptic curves, which have enhanced mathematical structure, leading to computational problems assumed to be hard both for classical and quantum computers. This framework is known as post-quantum cryptography, which represents a critical advancement in information security, ensuring that digital assets remain protected against evolving technological threats.

The Fiat--Shamir transform can be extended to lattices, leading to the Fiat--Shamir with aborts paradigm due to Lyubashevsky \cite{lyubashevsky2009}. Fiat--Shamir with aborts provides a framework for constructing digital signature schemes with provable security. This methodology addresses the challenge of generating signatures from lattice-based one-way functions by introducing a controlled rejection sampling technique, aborting. When a potential signature might reveal information about the secret key, the signing algorithm simply aborts and restarts the process. This paradigm converts an interactive identification protocol with a non-negligible probability of aborting into a signature scheme through iterative execution until a successful completion occurs, the aborting procedure. This transformation eliminates the need for interaction by substituting the verifier's challenge with a hash function evaluation, which security analyses treat as a random oracle.

Amongst the main constructions based on the Fiat--Shamir with aborts mechanism, we find Dilithium \cite{ducas2017} and HAETAE \cite{cheon2024}. Dilithium derives its security from the difficulty of solving certain lattice problems, specifically the Module Learning with Errors and Module Short Integer Solution problems. The importance of the scheme comes from the effective balance between security, signature size, and computational efficiency, making it practical for real-world implementations. HAETAE has been specifically designed to produce more compact and efficiently maskable signatures. While built upon the Fiat--Shamir with aborts paradigm underpinning Dilithium, HAETAE introduces design choices that optimize the complexity-to-compactness ratio, which is particularly important in space-constrained implementation scenarios.

Digital signatures are designed to provide authenticity, integrity, and non-repudiation for signed messages. Once a signature is created, it remains valid indefinitely, and the signer cannot rescind it. This permanence, while typically considered a feature, raises an important question: is it possible for signers to efficiently revoke their signatures without compromising their private keys or affecting the validity of their other signatures? Liu, Baek, and Susilo address this challenge by introducing the concept of withdrawable signatures \cite{liu2023}, offering a practical and secure mechanism for signature revocation in situations where this capability is desirable.

The practical applications of withdrawable signatures span multiple domains where signature revocation capability is essential without compromising the signer's private key. In blockchain-based smart contracts, these signatures enable participants to commit to contract conditions while maintaining the ability to revoke their commitment, particularly valuable when contract fulfillment depends on multiple parties or external conditions. Within decentralized e-voting systems, withdrawable signatures provide voters the security to cast their votes while retaining the flexibility to modify their choices before final vote tabulation, allowing voters to respond to new information or developments during the voting period. Additionally, in decentralized escrow services, these signatures facilitate multi-party transactions by allowing participants to revoke their signatures if circumstances change or disputes arise, without compromising the security of other parties' signatures.

\subsection{Contributions}
This paper builds upon the research conducted by Liu, Baek, and Susilo on withdrawable signatures \cite{liu2023} to present a general methodology for constructing post-quantum withdrawable signatures based on Lyubashevsky's Fiat--Shamir with aborts paradigm \cite{lyubashevsky2009}.

The work first introduces a comprehensive abstract construction that takes a mechanism based on the Fiat--Shamir with aborts paradigm as its starting point and extends it to incorporate withdrawability features. To demonstrate the practical applications of this abstract framework, the paper presents a concrete withdrawable construction from a no-hint, full-$\bt$ Dilithium-style variant \cite{ducas2017}; the no-hint variant keeps the naHVZK proof exact (perfect, $\varepsilon_{\mathrm{zk}}=0$), whereas production ML-DSA requires handling hints and incurs a small $\varepsilon_{\mathrm{zk}}$.

The integration of withdrawability features with post-quantum cryptography addresses two crucial challenges in modern cryptographic systems. First, it provides the flexibility to revoke signatures when needed, a capability increasingly important in dynamic digital environments. Second, it ensures this functionality remains secure against quantum computing threats.

\subsection{Related work}
The research in this paper revolves around three gravity centres, namely: the work done by Lyubashevsky, the construction of Liu, Baek and Susilo, and the digital scheme Dilithium.

The Fiat--Shamir heuristic \cite{fiat2000} provides a method for transforming public-coin interactive proof systems into digital signatures. This transformation works by substituting the verifier's public coin tosses with hash function evaluations. In the random oracle model (ROM), these hash functions are treated as uniform functions that adversaries can access through classical computing methods.

Schnorr's signature scheme stands as a prominent implementation of the Fiat--Shamir heuristic, with its security fundamentally based on the discrete logarithm problem. However, the emergence of quantum computing necessitates two critical adaptations to this framework. First, the discrete logarithm hardness assumption must be replaced with quantum-resistant alternatives. Second, the security model must account for quantum access to the random oracle (QROM), as quantum adversaries can query the hash function in superposition.

Lyubashevsky's work \cite{lyubashevsky2009} introduced an innovative lattice-based signature scheme that builds upon Schnorr's design while incorporating abortion as a crucial modification. The abort mechanism ensures that the signature distribution remains independent of the signing key, preventing potential attacks against the signature scheme. The protocol manages these aborts through a loop structure, continuing iterations until a successful execution occurs without an abort.

This modified approach, known as Fiat--Shamir with aborts, maintains the fundamental concept of replacing non-final verifier steps with hash function evaluations while adapting to the requirements of lattice-based cryptography. The integration of the abort mechanism represents a significant advancement in developing quantum-resistant signature schemes, providing a robust framework for cryptographic security in the post-quantum era.

On the other hand, Liu, Baek and Susilo propose a designated-verifier signature scheme that introduces withdrawability to digital signatures. Their construction generates a withdrawable signature $\sigma$ for a message $\mu$, distinct from conventional signature schemes.

The scheme's construction centers on a transformation mechanism. When a signer generates a withdrawable signature, it remains verifiable only by the designated verifier. The signature can then follow two paths: it either remains in its withdrawable state through signer inaction, effectively withdrawing the signature, or undergoes transformation through a ``Confirm'' algorithm. This algorithm converts the withdrawable signature $\sigma$ into a confirmed signature $\tilde{\sigma}$, which becomes verifiable through both parties' public keys while maintaining a deterministic relationship to $\sigma$.

The formal construction involves two entities, signer and verifier, with their public keys comprising a set $\pi=\{pk_s, pk_v\}$, respectively. The scheme utilizes the underlying signature structure to construct a withdrawable signature $\sigma$ specifically designated for the verifier. Subsequently, using the secret key $sk_s$ and $\sigma$, the signer can generate a verifiable signature for $\mu$ through the public key set $\pi$. This confirmed signature $\tilde{\sigma}$ maintains a cryptographic link to the original withdrawable signature $\sigma$ via the public key set $\pi$, ensuring the signature transformation's verifiability and traceability.

Finally, we find Dilithium \cite{ducas2017}, which builds upon the Fiat--Shamir with aborts paradigm. The scheme operates over module lattices, specifically using the module learning with errors (MLWE) and the module short integer solution (MSIS) problems as its primary security foundations. Its construction employs a ring $R_q=\mathbb{Z}_q[x]/(x^n+1)$ with carefully chosen parameters $n$ and $q$. The signature scheme utilizes matrices and vectors over this ring, with dimensions selected to balance security and efficiency.

The key generation in Dilithium produces a public key containing a matrix $\bA\in\Rq^{k\times l}$ and a vector $\bt=\bA\bs_1+\bs_2$, where $\bs_1$ and $\bs_2$ are secret vectors with small coefficients. The signing process involves generating a masking vector $\by$, computing $\bw=\bA\by$, and using a challenge value $c$ derived from the message and $\bw$ to produce the signature. The scheme incorporates rejection sampling to ensure signature security, leading to probabilistic signature generation.

\section{Preliminaries}
\subsection{Notation}
We write $\Ring$ and $\Rq$ to denote the rings $\mathbb{Z}[x]/(x^n+1)$ and $\mathbb{Z}_q[x]/(x^n+1)$ respectively, where $q$ is an integer. We will denote vectors by bold letters, and regular font letters denote elements in $\Ring$ or $\Rq$. Unless otherwise specified, all vectors are assumed to be column vectors. Bold upper-case letters denote matrices.

If $S$ is a set, then $a\sample S$ means that $a$ is chosen uniformly at random from $S$. All logarithms are assumed to be base $2$, and the hash function $H(\cdot,\cdot)$ will operate on the concatenation of its arguments.

For any even positive integer $\alpha$, we define $r'=r\modpm\alpha$ to be the unique element $r'$ in the range $-\tfrac{\alpha}{2}<r'\le\tfrac{\alpha}{2}$ such that $r'\equiv r \bmod \alpha$. For any odd positive integer $\alpha$, we define $r'=r\modpm\alpha$ to be the unique element $r'$ in the range $-\tfrac{\alpha-1}{2}\le r'\le\tfrac{\alpha-1}{2}$ such that $r'\equiv r\bmod\alpha$. We refer to these operations as centered reductions modulo $\alpha$.

For any positive integer $\alpha$, we define $r'=r\bmod^{+}\alpha$ to be the unique element $r'$ in the range $0\le r'<\alpha$ such that $r'\equiv r\bmod\alpha$. When the specific representation is not significant, we simply write $r\bmod\alpha$.

For an element $w\in\mathbb{Z}_q$, we denote $\norminf{w}$ to mean $|w\modpm q|$. For $w=\sum_{i=0}^{n-1}w_i x^i\in\Ring$, we define $\norminf{w}=\max_i\norminf{w_i}$ and $\norm{w}=\sqrt{\norminf{w_0}^2+\dots+\norminf{w_{n-1}}^2}$. For $\bw=(w_1,\dots,w_k)\in\Ring^k$: $\norminf{\bw}=\max_i\norminf{w_i}$ and $\norm{\bw}=\sqrt{\norm{w_1}^2+\dots+\norm{w_k}^2}$. We denote by $S_\alpha$ the set of all elements $w\in\Ring$ such that $\norminf{w}\le\alpha$ for $\alpha\in\Ring$. There will be an abuse of notation: we write $S$, with no subindex, to indicate that bounds in this case are given by the underlying scheme.

From \cite{ducas2017} we set $\beta$ as the maximum possible coefficient of $c\bs_i$, $\gamma_1$ is large enough so the signature does not reveal the secret key and small enough so that the signature is not forged, $\gamma_2$ is the low-order rounding bound with $2\gamma_2\mid q-1$ (concretely $\gamma_2=(q-1)/88$ for level 2 and $(q-1)/32$ for levels~3 and~5, so $\gamma_2\neq\gamma_1/2$ in general), and $\eta$ is a small integer. We write $\mathcal{C}=B_\tau$ for the challenge set of weight $\tau$ (Definition~\ref{def:stmsis}); the concrete weights are $\tau\in\{39,49,60\}$ for the three parameter sets, and $M=|\mathcal{C}|=2^{\tau}\binom{n}{\tau}$. Fix an efficiently invertible bijection $\iota:\mathbb{Z}_M\to\mathcal{C}$ (for instance the indexing underlying \textsf{SampleInBall}), with inverse $\iota^{-1}:\mathcal{C}\to\mathbb{Z}_M$. We use a challenge hash $H_M:\{0,1\}^*\to\mathbb{Z}_M$. Let $\mathsf{HB}(\cdot)$ denote the recoverable (high-order) part of a commitment, i.e.\ for every accepted transcript $\mathsf{HB}(\bA\bz-c\bt)=\mathsf{HB}(\bA\by)$; in Dilithium $\mathsf{HB}(\cdot)=\HighBits(\cdot,2\gamma_2)$.

Let $\PKE=(\mathsf{KGen}_e,\mathsf{Enc},\mathsf{Dec})$ be a public-key encryption scheme satisfying \emph{ciphertext pseudorandomness}: for honestly generated $ek$ and any message $K$, $\mathsf{Enc}(ek,K)$ is computationally indistinguishable from a uniform element of the ciphertext space given $ek$ (but not $dk$); we write its advantage $\Adv^{\mathrm{pr}}_{\PKE}$. Ciphertext pseudorandomness implies both IND- CPA and key-privacy (anonymity), but neither of those implies it. We use $\PKE$ to encrypt a chosen symmetric key $K$. The CPA-secure module-LWE encryption underlying ML-KEM (Kyber) is a candidate instantiation: its ciphertext components are pseudorandom under decisional MLWE, and pseudorandomness is understood with respect to the actual (compressed) ciphertext encoding space, up to the standard compression/statistical terms. Let $G:\{0,1\}^\kappa\to \{0,1\}^*$ be a pseudorandom generator (XOF). Each party holds, in addition to its lattice keys, an encryption key pair $(ek,dk)\leftarrow\mathsf{KGen}_e(1^\kappa)$; for a key set $\pi=\{pk_s,pk_v\}$ we write $ek_s,ek_v$ (resp.\ $dk_s,dk_v$) for the signer's and verifier's encryption keys.

Public parameters $pp\leftarrow\Setup(1^\kappa)$ fix the common matrix $\bA\in\Rq^{k\times l}$ (with the $\PKE$ and PRG descriptions); every key pair is generated as $(pk,sk)\leftarrow\KeyGen(pp)$ relative to the same $\bA$, so the two OR-proof branches share a module matrix.

\subsection{Basic definitions}
A withdrawable signature scheme involves two participating parties: signers and verifiers. The scheme operates in two primary stages: first, the generation of a withdrawable signature, and second, its transformation into a confirmed signature. Both stages are executed by the signer. Concerning the security of a withdrawable signature scheme $WS$, it is established through three properties, namely: correctness, unforgeability under insider corruption, and withdrawability.
\begin{enumerate}
  \item Correctness establishes a strong relation between the verification algorithms: if a withdrawable signature $\sigma$ is successfully verified through the $\WSVerify$ algorithm, then its corresponding confirmed signature $\tilde{\sigma}$ must also be verifiable through the $\CVerify$ algorithm.
  \item Unforgeability under insider corruption ensures that only the original signer possesses the capability to transform a verifiable withdrawable signature $\sigma$ (generated using $sk_s$ for verifier $pk_v$) into its corresponding confirmed signature $\tilde{\sigma}$. This requirement holds even when an adversary has obtained the verifier's secret key $sk_v$, maintaining the exclusive control of the signer over the confirmation process.
  \item Finally, withdrawability establishes the indistinguishability of signature origin. Specifically, given a verifiable withdrawable signature $\sigma$, no PPT adversary $\adv$ should be able to determine whether the signature was generated by the signer or the verifier, provided that the $\Confirm$ algorithm has not been executed on $\sigma$. This property effectively ensures that both the signer and the designated verifier possess equivalent capabilities in generating withdrawable signatures.
\end{enumerate}

\paragraph{Threat model.} All results are in the (quantum) random-oracle model, stated for the classical ROM with the QROM variant noted in Remark~\ref{rem:qrom}. Keys are honestly generated unless stated otherwise: withdrawability is proved for honest keys (Theorem~\ref{thm:withdraw}), and Remark~\ref{rem:wd} treats an adversarially chosen verifier key at the cost of a decisional-MLWE term. Unforgeability (Theorem~\ref{thm:unforge}) is under insider corruption, i.e.\ the adversary may hold the designated verifier's secret key. For withdrawability we distinguish third-party adversaries (holding no secret key), defeated by the encryption layer of Section~\ref{sec:sdvs}, from the designated verifier (holding $dk_v$), defeated by the origin ambiguity of the public core.

\begin{definition}[\cite{liu2023}, Section \ref{subsec:instantiation}]\label{def:ws}
A withdrawable signature scheme is composed of five polynomial-time algorithms $\KeyGen,\WSign,\WSVerify,\Confirm,\CVerify$, defined as follows:
\begin{itemize}
  \item $pp\leftarrow\Setup(1^\kappa)$, $(pk,sk)\leftarrow\KeyGen(pp)$: $\Setup$ fixes public parameters $pp$, in particular the common matrix $\bA$ shared by all users, and $\KeyGen(pp)$ outputs one key pair. The signer and verifier keys $(pk_s,sk_s)$, $(pk_v,sk_v)$ are two independent invocations of $\KeyGen(pp)$ and thus share $\bA$; by convention each secret key embeds its own public key, i.e.\ $sk=(pk,\bs_1,\bs_2,dk)$, so that a party can reconstruct $\pi$ from its own $sk$.
  \item $\sigma\leftarrow\WSign(\mu,sk_s,\pi)$: given a message $\mu$, a signer's secret key $sk_s$, and a tuple $\pi=\{pk_s,pk_v\}$ containing both the signer's public key $pk_s$ and a designated verifier's public key $pk_v$ from the set of all public keys $PK$ (those output by $\KeyGen$), the withdrawable signing algorithm generates a withdrawable signature $\sigma$. This signature is specifically bound to message $\mu$ under the signer's identity and can only be verified by the designated verifier $pk_v$.
  \item $1/0\leftarrow\WSVerify(\mu,sk_v,pk_s,\sigma)$: the withdrawable signature verification algorithm uses the designated verifier's secret key $sk_v$ (which includes $dk_v$) to return $1$ iff $\sigma$ is a valid withdrawable signature on $\mu$ under $\pi=\{pk_s,pk_v\}$, and $0$ otherwise. Among third parties only the designated verifier can run it; the signer may also recover the public core (via $dk_s$) for confirmation. The signature is not publicly verifiable.
  \item $\tilde{\sigma}\leftarrow\Confirm(\mu,sk_s,\pi,\sigma)$: the confirmation algorithm transforms a withdrawable signature $\sigma$ into a confirmed signature $\tilde{\sigma}$. It takes as input the original message $\mu$, the signer's secret key $sk_s$, the public key set $\pi$, and the withdrawable signature $\sigma$. The resulting confirmed signature $\tilde{\sigma}$ serves as a publicly verifiable signature with respect to the key set $\pi$.
  \item $1/0\leftarrow\CVerify(\mu,\pi,\sigma,\tilde{\sigma})$: the confirmed signature verification algorithm validates the authenticity of a confirmed signature $\tilde{\sigma}$ on message $\mu$ with respect to the public key set $\pi$. It takes as additional input the original withdrawable signature $\sigma$ from which the confirmed signature was derived. The algorithm outputs $1$ if the confirmed signature is valid and $0$ otherwise.
\end{itemize}
\end{definition}

\begin{definition}\label{def:correct}
A withdrawable signature scheme $WS$ is correct if, for any security parameter $\kappa$, any public key set $\pi$, and any message $\mu\in\{0,1\}^*$, when executing the sequence $\KeyGen,\WSign$, and $\Confirm$, then the corresponding verification algorithms satisfy:
\[
\WSVerify(\mu,sk_v,pk_s,\sigma)=1 \quad\text{and}\quad \CVerify(\mu,\pi,\sigma,\tilde{\sigma})=1
\]
with overwhelming probability (in the security parameter $\kappa$).
\end{definition}

\begin{definition}\label{def:euf}
For a PPT adversary $\adv$ and security parameter $\kappa$, we define the unforgeability under insider corruption experiment $\Exp^{\text{EUF-CMA}}_{WS,\adv}(1^\kappa)$ using the three oracles in Algorithms \ref{alg:corrupt}--\ref{alg:confirm-oracle}.
\end{definition}

\begin{algorithm}
\caption{Corruption oracle $\Ora{Corrupt}_i(\cdot)$}\label{alg:corrupt}
\begin{algorithmic}[1]
\If{$i\neq s$} $CO\leftarrow CO\cup\{i\}$
  \State \textbf{return} $sk_i$
\Else\ \textbf{return} $\bot$
\EndIf
\end{algorithmic}
\end{algorithm}

\begin{algorithm}
\caption{Withdrawable signing oracle $\Ora{WSign}_{sk_s,\pi}(\mu)$}\label{alg:wsign-oracle}
\begin{algorithmic}[1]
\If{$(pk_s\in\pi)\wedge(s\notin CO)$}
  \State $\sigma\leftarrow\WSign(\mu,sk_s,\pi)$
  \State $W\leftarrow W\cup\{\sigma\}$
  \State \textbf{return} $\sigma$
\Else\ \textbf{return} $\bot$
\EndIf
\end{algorithmic}
\end{algorithm}

\begin{algorithm}
\caption{Confirmation oracle $\Ora{Confirm}_{sk_s,\pi}(\mu,\sigma)$}\label{alg:confirm-oracle}
\begin{algorithmic}[1]
\If{$\sigma\in W$}
  \State $M\leftarrow M\cup\{\mu\}$
  \State $\tilde{\sigma}\leftarrow\Confirm(\mu,sk_s,\pi,\sigma)$
  \State \textbf{return} $\tilde{\sigma}$
\Else\ \textbf{return} $\bot$
\EndIf
\end{algorithmic}
\end{algorithm}

Using these three oracles, we define the unforgeability experiment $\Exp^{\text{EUF-CMA}}_{WS,\adv}(1^\kappa)$ as follows:

\begin{algorithm}
\caption{Unforgeability experiment $\Exp^{\text{EUF-CMA}}_{WS,\adv}(1^\kappa)$}\label{alg:euf-exp}
\begin{algorithmic}[1]
\State $pp\leftarrow\Setup(1^\kappa)$
\For{$i=1$ to $m$} $(pk_i,sk_i)\leftarrow\KeyGen(pp)$
\EndFor
\State Select $s,v\in\{1,\dots,m\}$ where $v\neq s$
\State Initialize empty sets $CO\leftarrow\emptyset,\ W\leftarrow\emptyset,\ M\leftarrow\emptyset$
\State $(\mu^*,\sigma^*,\tilde{\sigma}^*)\leftarrow\adv^{\Ora{Corrupt}_i(\cdot),\,\Ora{WSign}_{sk_s,\pi}(\cdot),\,\Ora{Confirm}_{sk_s,\pi}(\cdot,\cdot)}(1^\kappa,\pi^*)$
\If{$(\pi^*=\{pk_s,pk_v\})\wedge(v\in CO)\wedge(\mu^*\notin M)$}
  \If{$(\WSVerify(\mu^*,sk_v,pk_s,\sigma^*)=1)\wedge(\CVerify(\mu^*,\pi^*,\sigma^*,\tilde{\sigma}^*)=1)$} \textbf{return} $1$
  \EndIf
\EndIf
\State \textbf{return} $0$
\end{algorithmic}
\end{algorithm}

A withdrawable signature scheme $WS$ is considered unforgeable under insider corruption with EUF-CMA security if, for all PPT adversaries $\adv$, there exists a negligible function $\negl$ such that:
\[
\Pr[\Exp^{\text{EUF-CMA}}_{WS,\adv}(1^\kappa)=1]\le\negl(1^\kappa).
\]

\begin{definition}\label{def:withdraw}
Let $(pk_0,sk_0),(pk_1,sk_1)\leftarrow\KeyGen(pp)$ be two generated public/secret key pairs, and let $\pi=\{pk_0,pk_1\}$. For a randomly selected bit $b\sample\{0,1\}$, a security parameter $\kappa$, and a PPT adversary $\adv$, we build the withdrawability experiment $\Exp^{\text{Withdraw}}_{WS,\adv}(1^\kappa)$ with the oracle of Algorithm \ref{alg:withdraw-oracle}.
\end{definition}

\begin{algorithm}
\caption{Signing oracle $\Ora{WSign}_{sk_b,\pi}(\cdot)$ for the withdrawability experiment}\label{alg:withdraw-oracle}
\begin{algorithmic}[1]
\If{$\pi=\{pk_0,pk_1\}$}
  \State $\sigma_b\leftarrow\WSign(\mu,sk_b,\pi)$
  \State $M\leftarrow M\cup\{\mu\}$
  \State \textbf{return} $\sigma_b$
\Else\ \textbf{return} $\bot$
\EndIf
\end{algorithmic}
\end{algorithm}

With this signing oracle, we have the following experiment:

\begin{algorithm}
\caption{Withdrawability experiment $\Exp^{\text{Withdraw}}_{WS,\adv}(1^\kappa)$}\label{alg:withdraw-exp}
\begin{algorithmic}[1]
\State $pp\leftarrow\Setup(1^\kappa)$;\ \textbf{for }$i=0,1$: $(pk_i,sk_i)\leftarrow\KeyGen(pp)$
\State $\pi\leftarrow\{pk_0,pk_1\};\ b\sample\{0,1\};\ M\leftarrow\emptyset$
\State $(\mu^*,\mathit{st})\leftarrow\adv^{\Ora{WSign}_{sk_b,\pi}}(pp,\pi)$ \Comment{stage 1: choose
challenge message}
\State $\sigma_b\leftarrow\WSign(\mu^*,sk_b,\pi)$
\State $b'\leftarrow\adv^{\Ora{WSign}_{sk_b,\pi}}(\mathit{st},\sigma_b)$ \Comment{stage 2: guess}
\If{$(b=b')\wedge(\mu^*\notin M)$} \textbf{return} $1$
\EndIf
\State \textbf{return} $0$
\end{algorithmic}
\end{algorithm}

A withdrawable signature scheme $WS$ is withdrawable if, for any PPT adversary $\adv$, and in the absence of the execution of the $\Confirm$ algorithm, there exists a negligible function $\negl$ such that:
\[
\Pr[\Exp^{\text{Withdraw}}_{WS,\adv}(1^\kappa)=1]\le\tfrac{1}{2}+\negl(1^\kappa).
\]

\subsection{Security definitions and computational assumptions}

\begin{definition}\label{def:ds-euf}
For a signature scheme $DS=(\KeyGen,\Sign,\Verify)$ and a PPT adversary $\adv$, consider the following experiment $\Exp^{\text{EUF-CMA}}_{\adv}$:
\begin{enumerate}
  \item The challenger $\mathcal{B}$ generates a key pair $(pk_s,sk_s)\leftarrow\KeyGen(pp)$ using the system parameters $SP$. It provides $pk_s$ to $\adv$ while retaining $sk_s$ to handle signature queries.
  \item $\adv$ receives access to the signing oracle $\Ora{Sign}_{sk_s}(\cdot)$ that computes $\sigma\leftarrow\Sign(\mu,sk_s)$ upon request.
  \item Eventually, $\adv$ outputs a forgery attempt $(\mu^*,\sigma^*)$.
  \item $\adv$ succeeds if $\Verify(\mu^*,pk_s,\sigma^*)=1$ and $\mu^*$ was not previously queried to $\Ora{Sign}_{sk_s}(\cdot)$.
\end{enumerate}
We say that $DS$ is $(t,q_s,\varepsilon)$-secure under EUF-CMA if no adversary running in time $t$ and making at most $q_s$ signing queries can succeed with probability greater than $\varepsilon$.
\end{definition}

\begin{definition}\label{def:dvs}
A designated-verifier signature scheme $DVS$ consists of four probabilistic polynomial-time algorithms operating on key pairs $(pk_s,sk_s)$ for signers and $(pk_d,sk_d)$ for designated verifiers:
\[
DVS=\left\{
\begin{aligned}
&(pk,sk)\leftarrow\KeyGen(pp)\\
&\sigma\leftarrow\Sign(\mu,pk_d,sk_s)\\
&\sigma\leftarrow\Simul(\mu,pk_s,sk_d)\\
&0/1\leftarrow\Verify(\mu,pk_s,sk_d,\sigma)
\end{aligned}
\right\}
\]
The key security property of $DVS$ schemes is non-transferability, which states that for any message-signature pair $(\mu,\sigma)$ that validates under $\Verify$, it should be computationally infeasible to determine whether $\sigma$ was produced by the signer using $\Sign$ or simulated by the designated verifier using $\Simul$, without access to the signer's secret key $sk_s$. The formal definition of this property follows:
\end{definition}

\begin{definition}[Non-transferability]\label{def:nontrans}
For a designated-verifier signature scheme and a PPT adversary $\adv$, consider the non-transferability experiment $\Exp^{\Sign}_{\text{NonTrans},DV,\adv}$:
\end{definition}

\begin{algorithm}
\caption{Non-transferability experiment $\Exp^{\Sign}_{\text{NonTrans},DV,\adv}(1^\kappa)$}\label{alg:nontrans}
\begin{algorithmic}[1]
\State $(pk_s,sk_s),(pk_d,sk_d)\leftarrow\KeyGen(pp)$
\State Provide $\adv$ access to oracles:
\State $\Ora{Sign}_{sk_s,pk_d}(\cdot):\ \sigma_0\leftarrow\Sign(\mu,sk_s,pk_d)$
\State $\Ora{Simul}_{sk_d,pk_s}(\cdot):\ \sigma_1\leftarrow\Simul(\mu,pk_s,sk_d)$
\State $\adv$ outputs message $\mu^*$
\State $\sigma^*_0\leftarrow\Sign(\mu^*,pk_d,sk_s);\quad \sigma^*_1\leftarrow\Simul(\mu^*,pk_s,sk_d)$
\State $b\sample\{0,1\}$
\State Provide $\adv$ with signature $\sigma^*_b$
\State $\adv$ outputs bit $b'$
\If{$b'=b$} \textbf{return} $1$
\Else\ \textbf{return} $0$
\EndIf
\end{algorithmic}
\end{algorithm}

A $DVS$ achieves non-transferability if for any PPT adversary $\adv$, there exists a negligible function $\negl$ such that:
\[
\Pr[\Exp^{\Sign}_{\text{NonTrans},DV,\adv}(1^\kappa)=1]\le\tfrac{1}{2}+\negl(1^\kappa).
\]

The security of our scheme rests upon three fundamental lattice-based hardness assumptions. The Module Learning With Errors (MLWE) assumption provides protection against key-recovery attacks, ensuring the confidentiality of secret keys. The SelfTargetMSIS assumption establishes the security foundation against new message forgery attempts, preventing adversaries from generating valid signatures for previously unsigned messages. Finally, the MSIS assumption is essential for achieving strong unforgeability, which prevents even slight modifications of existing signatures.

\begin{definition}[Module Learning With Errors (MLWE)]\label{def:mlwe}
For integers $m,k$ and a probability distribution $D:\Rq\to[0,1]$, the advantage of an algorithm $\adv$ in solving the decisional $\text{MLWE}_{m,k,D}$ problem over the ring $\Rq$ is defined as:
\begin{align*}
\Adv^{\text{MLWE}}_{m,k,D}:=\big|&\Pr[b=1\mid \bA\leftarrow\Rq^{m\times k};\ \bt\leftarrow\Rq^m;\ b\leftarrow\adv(\bA,\bt)]\\
-&\Pr[b=1\mid \bA\leftarrow\Rq^{m\times k};\ \bs_1\leftarrow D^k;\ \bs_2\leftarrow D^m;\ b\leftarrow\adv(\bA,\bA\bs_1+\bs_2)]\big|.
\end{align*}
\end{definition}

\begin{definition}[Module Short Integer Solution (MSIS)]\label{def:msis}
For an algorithm $\adv$, we define its advantage function $\Adv^{\text{MSIS}}_{m,k,\gamma}$ in solving the (Hermite Normal Form) $\text{MSIS}_{m,k,\gamma}$ problem over the ring $\Rq$ as:
\[
\Adv^{\text{MSIS}}_{m,k,\gamma}(\adv):=\Pr\!\left[0<\norminf{\by}\le\gamma\ \wedge\ [\,\bI\mid\bA\,]\cdot\by=0\ \middle|\ \bA\leftarrow\Rq^{m\times k};\ \by\leftarrow\adv(\bA)\right].
\]
\end{definition}

Let $B_h$ denote the subset of elements in $\Ring$ that have exactly $h$ coefficients equal to either $-1$ or $1$, with all remaining coefficients being $0$. The cardinality of this set is given by $|B_h|=2^h\cdot\binom{n}{h}$.

\begin{definition}[SelfTargetMSIS]\label{def:stmsis}
Let $H:\{0,1\}^*\to B_h$ be a cryptographic hash function. For an algorithm $\adv$, we define its advantage function as:
\[
\Adv^{\text{SelfTargetMSIS}}_{H,m,k,\gamma_1}(\adv):=\Pr\!\left[
\begin{aligned}
&0\le\norminf{\by}\le\gamma_1\ \wedge\ H\big([\,\bI\mid\bA\,]\cdot\by\,\big\|\,M\big)=c\\
&\ \big|\ \bA\leftarrow\Rq^{m\times k};\ (\by:=(\mathbf{r},c),M)\leftarrow\adv^{\langle H(\cdot)\rangle}(\bA)
\end{aligned}
\right].
\]
\end{definition}

The above problems are hard if, for any PPT adversary $\adv$, the respective advantages are negligible.

The intuition is that the MLWE assumption protects against key-recovery attacks, the SelfTargetMSIS is the assumption upon which new message forgery is based, and the MSIS assumption is needed for strong unforgeability. The decisional MLWE will be used to prove withdrawability.

\section{Withdrawable signature schemes}
\subsection{The abstract construction}\label{sec:abstract}
In this section we provide a global construction for withdrawable signatures based on the Fiat--Shamir with aborts paradigm and prove the unforgeability and withdrawability as described in Definition \ref{def:ws}. This construction follows the paradigm described in \cite{lyubashevsky2009} combined with ideas in \cite{liu2023}. Let $\mathcal{C}=B_h$ be the (bounded) challenge set and $H:\{0,1\}^*\to\mathcal{C}$ a random oracle. Sampling uniformly from $\Ring$ is not well defined, and the challenge must be small enough that $\bz=\by+c\bs_1\in S^l$.

\begin{algorithm}
\caption{Fiat--Shamir with aborts signature}\label{alg:fswa}
\begin{algorithmic}[1]
\Procedure{KeyGen}{$1^\kappa$}
  \State $\bA\leftarrow\Rq^{k\times l},\ \bs_1\leftarrow S^l,\ \bs_2\leftarrow S^k$
  \State $\bt=\bA\bs_1+\bs_2$
  \State \textbf{return} $pk=(\bA,\bt),\ sk=(\bs_1,\bs_2)$
\EndProcedure
\Procedure{Sign}{$\mu,sk,\bt$}
  \Repeat\ $\by\leftarrow S^l,\ \bw=\mathsf{HB}(\bA\by),\ c=H(\mu,\bw),\ \bz=\by+c\bs_1$
  \Until{$\bz\in S^l$ \textbf{and} the recoverability bound holds}
  \State \textbf{return} $\sigma=(\bz,c)$
\EndProcedure
\Procedure{Verify}{$\mu,\sigma,pk$}
  \State $\bw'=\mathsf{HB}(\bA\bz-c\bt)$
  \If{$(\bz\in S^l)\wedge(c=H(\mu,\bw'))$} \textbf{return} $1$
  \EndIf
\EndProcedure
\end{algorithmic}
\end{algorithm}

\begin{theorem}[\cite{lyubashevsky2009}, Theorem 2]\label{thm:lyu}
Let $n$ be an integer which is a power of $2$. If the above signature scheme is not strongly unforgeable, then there is a polynomial-time algorithm that can solve $\SVP_\varepsilon(\Lambda)$, for $\varepsilon=\tilde{O}(n^2)$ and for every lattice $\Lambda$ corresponding to an ideal in the ring $\mathbb{Z}[x]/(x^n+1)$.
\end{theorem}

Taking the above scheme as starting point, we define a withdrawable lattice-based scheme as follows:

\begin{algorithm}
\caption{$\KeyGen(pp)$}\label{alg:abs-keygen}
\begin{algorithmic}[1]
\State parse $\bA$ from $pp$
\State $\bs_1\leftarrow S^l,\ \bs_2\leftarrow S^k$
\State $\bt=\bA\bs_1+\bs_2$
\State $(ek,dk)\leftarrow\mathsf{KGen}_e(1^\kappa)$
\State \textbf{return} $pk=(\bA,\bt,ek),\ sk=(pk,\bs_1,\bs_2,dk)$
\end{algorithmic}
\end{algorithm}

\begin{algorithm}
\caption{$\WSign^{\mathsf{pub}}(\mu,sk_s,\pi)$}\label{alg:abs-wsign}
\begin{algorithmic}[1]
\State $\pi=\{pk_s,pk_v\}$
\Repeat \Comment{simulate verifier branch from the accepted distribution}
\State $c_v\sample\mathcal{C};\ \bz_v\sample S^l;\ w_v=\mathsf{HB}(\bA\bz_v-c_v\bt_v)$
\Until{the low-order recoverability test on $\bA\bz_v-c_v\bt_v$ holds}
\Repeat
  \State $\by\sample S^l;\ w_s=\mathsf{HB}(\bA\by)$
  \State $g=H_M(\mu,w_s,w_v,\pi);\ c_s=\iota\big((g-\iota^{-1}(c_v))\bmod M\big)$
  \State $\bz_s=\by+c_s\bs_1$ \Comment{$sk_s=(pk_s,\bs_1,\bs_2,dk_s)$}
\Until{$\bz_s\in S^l$ \textbf{and} the recoverability bound holds}
\State \textbf{return} $\sigma=(c_s,c_v,\bz_s,\bz_v)$
\end{algorithmic}
\end{algorithm}

\begin{algorithm}
\caption{$\WSVerify^{\mathsf{pub}}(\mu,\pi,\sigma)$}\label{alg:abs-wsverify}
\begin{algorithmic}[1]
\State $\sigma=(c_s,c_v,\bz_s,\bz_v)$
\State $w_s'=\mathsf{HB}(\bA\bz_s-c_s\bt_s);\ w_v'=\mathsf{HB}(\bA\bz_v-c_v\bt_v)$
\State $g=H_M(\mu,w_s',w_v',\pi)$
\If{$(\iota^{-1}(c_s)+\iota^{-1}(c_v)\equiv g \!\!\pmod{M})\wedge(\bz_s,\bz_v\in S^l)$} \textbf{return} $1$
\EndIf
\end{algorithmic}
\end{algorithm}

\begin{algorithm}
\caption{$\Confirm^{\mathsf{pub}}(\mu,sk_s,\pi,\sigma)$}\label{alg:abs-confirm}
\begin{algorithmic}[1]
\State \textbf{return} $\tilde{\sigma}\leftarrow\Sign_{sk_s}(\mu\,\|\,\pi\,\|\,\sigma)$
\end{algorithmic}
\end{algorithm}

\begin{algorithm}
\caption{$\CVerify^{\mathsf{pub}}(\mu,\pi,\sigma,\tilde{\sigma})$}\label{alg:abs-cverify}
\begin{algorithmic}[1]
\If{$\WSVerify^{\mathsf{pub}}(\mu,\pi,\sigma)=1\ \wedge\ \Verify_{pk_s}(\mu\,\|\,\pi\,\|\,\sigma,\ \tilde{\sigma})=1$} \textbf{return} $1$
\EndIf
\end{algorithmic}
\end{algorithm}

The proofs for Theorem \ref{thm:unforge} and Theorem \ref{thm:withdraw} below closely follow the structure of their analogous Theorems 12 and 13 in \cite{liu2023}. The fundamental arguments remain valid when adapting the security assumptions to our context, with the primary distinction being the underlying signature scheme.

\begin{assumption}[Commitment min-entropy]\label{ass:entropy}
The commitment pair $(w_s,w_v)$ produced by an accepted run of $\WSign$ has min-entropy at least $\alpha$; for the concrete Dilithium parameter sets $\alpha\ge 1024$ (and $\alpha\ge 1397$ at level~2), so $2^{-\alpha}\le 2^{-1024}$.
\end{assumption}

\begin{remark}[Heuristic support for Assumption~\ref{ass:entropy}]
Let $w_s[1]=\mathbf{a}_1\!\cdot\!\by$ be one output coordinate ($\mathbf{a}_1\in\Rq^{\,l}$ the first row of $\bA$). The input min-entropy $H_\infty(\by)=l\,n\log_2(2\gamma_1)$ exceeds $\log_2|\Rq|=n\log_2 q$ by $\ge 12544$ (level~2), $19712$ (level~3), $29952$ (level~5) bits, so under the module regularity of the SIS map $(\mathbf{a}_1,\mathbf{a}_1\!\cdot\!\by)$ is statistically close to $(\mathbf{a}_1,\mathcal U(\Rq))$, and $\HighBits(w_s[1],2\gamma_2)$ has min-entropy $\approx n\log_2((q-1)/2\gamma_2)\ge 1024$. Conditioning on the constant-probability rejection event lowers this by $O(1)$ bits. A fully rigorous bound needs a regularity lemma over $\Rq$ with exact conditions for a completely split $q$ and the rejection-conditioned law of $\by$; we therefore record the value as Assumption~\ref{ass:entropy}.
\end{remark}
 
\begin{lemma}[Witness-free simulatability / naHVZK]\label{lem:zk}
There is a PPT algorithm $\mathsf{Sim}(\mu,\pi)$, taking no secret key and fixing $H_M$ at a single point, such that for all honestly generated $\pi=\{pk_s,pk_v\}$ and all $\mu$, the output of $\mathsf{Sim}(\mu,\pi)$ together with its induced $H_M$-entry is within statistical distance $\varepsilon_{\mathrm{zk}}$ of a genuine pair $(\WSign(\mu,sk_s,\pi),H_M\text{-entry})$, conditioned on the programmed point being fresh. Concretely $\mathsf{Sim}$ repeatedly samples $c_s,c_v\sample\mathcal{C}$ and $\bz_s,\bz_v\sample S_{\gamma_1-\beta-1}^l$ until both $\norminf{\LowBits(\bA\bz_s-c_s\bt_s,2\gamma_2)}<\gamma_2-\beta$ and $\norminf{\LowBits(\bA\bz_v-c_v\bt_v,2\gamma_2)}<\gamma_2-\beta$ (the same low-order test as $\WSign$); it then sets $w_s=\HighBits(\bA\bz_s-c_s\bt_s,2\gamma_2)$, $w_v=\HighBits(\bA\bz_v-c_v\bt_v,2\gamma_2)$, programs $H_M(\mu,w_s,w_v,\pi):=(\iota^{-1}(c_s)+\iota^{-1}(c_v))\bmod M$, and returns $(c_s,c_v,\bz_s,\bz_v)$.
\end{lemma}

\begin{proof}
The verifier branch is produced by exactly this simulation already in $\WSign$, so it is identical in both distributions. The signer branch is one FS-with-aborts transcript; by the non-abort honest-verifier zero-knowledge of the base scheme its accepted output
$(c_s,\bz_s)$ is within statistical distance $\varepsilon_{\mathrm{zk}}$ of $(c_s\sample\mathcal{C},\ \bz_s\ \text{short},\ w_s=\HighBits(\bA\bz_s-c_s\bt_s,2\gamma_2))$. Finally the law of $(c_s,c_v,g)$ agrees: in the real scheme $c_v$ and $g=H_M(\cdots)$ are uniform and independent and $c_s=\iota((g-\iota^{-1}(c_v))\bmod M)$; in $\mathsf{Sim}$ both $c_s,c_v$ are uniform and $g$ is set to $\iota^{-1}(c_s)+\iota^{-1}(c_v)$. Both yield $(c_s,c_v)\sim\mathcal{U}(\mathcal{C})^2$ with $g$ determined by the constraint.
\end{proof}

\begin{theorem}\label{thm:unforge}
If the underlying signature scheme $\Sign$ is EUF-CMA secure, then the publicly-verifiable core of Algorithms \ref{alg:abs-keygen}--\ref{alg:abs-cverify} is unforgeable under insider corruption in the (Q)ROM (the full designated scheme of Section \ref{sec:sdvs} inherits this via Proposition \ref{prop:sdvs}). For any adversary $\adv$ making $q_W$ signing, $q_C$ confirmation and $q_{H_M}$ random-oracle queries, there is an adversary $\mathcal{B}$ such that, in the random-oracle model,
\[
\Adv^{\eufic}_{WS}(\adv)\ \le\ \Adv^{\eufcma}_{\Sign}(\mathcal{B})\ +\ q_W\,\varepsilon_{\mathrm{zk}}\ +\ \frac{q_W(q_{H_M}+q_W)}{2^{\alpha}},
\]
where $\varepsilon_{\mathrm{zk}}$ is the per-proof simulation error of Lemma~\ref{lem:zk} and $\alpha$ the commitment min-entropy (Assumption~\ref{ass:entropy}). The last term is the random-oracle programming loss; in the quantum random-oracle model it is replaced by the adaptive-reprogramming term of Remark~\ref{rem:qrom}. The reduction incurs no multiplicative $q_H$ factor.
\end{theorem}

\begin{remark}
Algorithms~\ref{alg:abs-keygen}--\ref{alg:abs-cverify} form the \emph{publicly-verifiable core}: Theorem~\ref{thm:withdraw} establishes its origin ambiguity and Theorem~\ref{thm:unforge} the unforgeability of its confirmation. The designated withdrawable signature of Definition~\ref{def:ws} is the strong designated-verifier scheme of Section~\ref{sec:sdvs}, which wraps this core and inherits both properties (Proposition~\ref{prop:sdvs}).    
\end{remark}

\begin{proof}
We use a short sequence of games; $\Pr[\mathsf{G}_i]$ is the probability that game $\mathsf{G}_i$ outputs $1$. Throughout, $H_M$ (the OR-proof challenge oracle) is domain-separated from the random oracle of the underlying scheme $DS$ of Definition \ref{def:ds-euf}; recall $\Confirm(\mu,sk_s,\pi,\sigma)=\Sign_{sk_s}(\mu\|\pi\|\sigma)$ and $\CVerify(\mu,\pi,\sigma,\tilde\sigma)=1$ iff $\WSVerify^{\mathsf{pub}}(\mu,\pi,\sigma)=1$ and $\Verify_{pk_s}(\mu\|\pi\|\sigma,\tilde\sigma)=1$, with $\langle\cdot\rangle$ an injective encoding.
 
\medskip\noindent\textbf{Game $\mathsf{G}_0$.} The experiment $\Exp^{\eufic}_{WS,\adv}(1^\kappa)$ of Definition \ref{def:euf} (Algorithm \ref{alg:euf-exp}). Thus $\Pr[\mathsf{G}_0]=\Adv^{\eufic}_{WS}(\adv)$.
 
\medskip\noindent\textbf{Game $\mathsf{G}_1$.} As $\mathsf{G}_0$, but $\Ora{WSign}$ is answered by $\mathsf{Sim}(\mu,\pi)$ of Lemma \ref{lem:zk}, programming the one induced $H_M$-entry; if that entry is already defined the game raises $\mathsf{bad}$ and aborts. By a hybrid over the $q_W$ signing queries, each replacement changes $\adv$'s view by at most $\varepsilon_{\mathrm{zk}}$ (Lemma \ref{lem:zk}) unless the programmed point $(\mu,w_s,w_v,\pi)$ already occurs in the $H_M$-table; at that moment the table has at most $q_{H_M}+q_W$ entries and, by Definition \ref{ass:entropy}, $(w_s,w_v)$ is unpredictable with min-entropy $\ge\alpha$, so a collision occurs with probability at most $(q_{H_M}+q_W)/2^{\alpha}$. Hence
\[
\big|\Pr[\mathsf{G}_0]-\Pr[\mathsf{G}_1]\big|\le q_W\,\varepsilon_{\mathrm{zk}} +\frac{q_W(q_{H_M}+q_W)}{2^{\alpha}}.
\]
 
In $\mathsf{G}_1$ the signer key $sk_s$ is used \emph{only} inside $\Ora{Confirm}$, via $\Sign_{sk_s}(\langle\mu,\pi,\sigma\rangle)$: $\Ora{Corrupt}$ never returns $sk_s$; $\Ora{WSign}$ now runs $\mathsf{Sim}$ (no secret key); the oracles $H,H_M$ are key-independent; and the keys of parties $i\neq s$ are independent of $sk_s$.
 
We build $\mathcal{B}$ against the $\eufcma$-security of $DS$ (Definition \ref{def:ds-euf}). $\mathcal{B}$ gets a challenge key $pk^\ast$, a signing oracle $\Sign_{sk^\ast}(\cdot)$ and access to $DS$'s random oracle $H$. It generates a fresh encryption key pair $(ek_s,dk_s)\leftarrow\mathsf{KGen}_e(1^\kappa)$ itself and sets $pk_s:=(pk^\ast,ek_s)$, so that $\Verify_{pk_s}$ parses and uses only the signing component $pk^\ast$ while $ek_s$ enters $\pi$ exactly as in an honest key; it never uses $dk_s$. It runs $(pk_i,sk_i)\leftarrow\KeyGen(pp)$ for all $i\neq s$ (so it holds $sk_v$), sets $\pi:=\{pk_s,pk_v\}$, and runs $\adv$, answering: $\Ora{Corrupt}(i)$ by $\bot$ if $i=s$ and by $sk_i$ otherwise (in particular it can deliver $sk_v$); $H_M$ by lazy sampling consistent with programming; $H$ by relaying to its own oracle (it never programs $H$); $\Ora{WSign}(\mu)$ by $\mathsf{Sim}(\mu,\pi)$ (aborting on $\mathsf{bad}$); and $\Ora{Confirm}(\mu,\sigma)$, for $(\mu,\sigma)\in W$, by setting $M\leftarrow M\cup\{\mu\}$ and returning $\tilde\sigma\leftarrow\Sign_{sk^\ast}(\langle\mu,\pi,\sigma\rangle)$. This is a perfect emulation of $\mathsf{G}_1$. When $\adv$ outputs $(\mu^\ast,\sigma^\ast,\tilde\sigma^\ast)$, $\mathcal{B}$ outputs $(\langle\mu^\ast,\pi,\sigma^\ast\rangle,\tilde\sigma^\ast)$.
 
If $\adv$ wins $\mathsf{G}_1$ then $\CVerify(\mu^\ast,\pi,\sigma^\ast,\tilde\sigma^\ast)=1$, hence $\Verify_{pk^\ast}(\langle\mu^\ast,\pi,\sigma^\ast\rangle,\tilde\sigma^\ast)=1$. The only messages $\mathcal{B}$ sent to its signing oracle are $\langle\mu_j,\pi,\sigma_j\rangle$ for the confirmation queries $j$, each of which placed $\mu_j$ into $M$. Since $\mu^\ast\notin M$, we have $\mu^\ast\neq\mu_j$ for all $j$, so by injectivity of $\langle\cdot\rangle$ the forgery message $\langle\mu^\ast,\pi,\sigma^\ast\rangle$ is new. Thus $\mathcal{B}$ wins, giving $\Pr[\mathsf{G}_1]\le\Adv^{\eufcma}_{DS}(\mathcal{B})$; $\mathcal{B}$ makes at most $q_C$ signing queries. Combining,
\[
\Adv^{\eufic}_{WS}(\adv)\le \Adv^{\eufcma}_{DS}(\mathcal{B}) +q_W\,\varepsilon_{\mathrm{zk}}+\frac{q_W(q_{H_M}+q_W)}{2^{\alpha}}.
\]
The reduction never rewinds $\adv$ and never extracts from the OR-proof (only its simulatability is used), so there is no multiplicative $q_H$ loss. Insider corruption is captured: $\mathcal{B}$ holds $sk_v$ and may hand it to $\adv$, yet confirmation still needs $sk_s=sk^\ast$, which $\mathcal{B}$ never knows.
\end{proof}

\begin{remark}[QROM]\label{rem:qrom}
In the QROM the only changes are: (i) the simulation of $\Ora{WSign}$ programs $H_M$ at an adaptively chosen, high-min-entropy point under superposition queries, which is handled by the adaptive reprogramming lemma of Grilo--H\"ovelmanns--H\"ulsing--Majenz \cite{grilo2021}, turning the term $q_W(q_{H_M}+q_W)/2^{\alpha}$ into $\tfrac{3}{2}q_W\sqrt{(q_{H_M}+q_W+1)/2^{\alpha}}$; and (ii) Definition \ref{def:ds-euf} is taken in the QROM, which for Dilithium follows from $\mathrm{SelfTargetMSIS}$ and $\mathrm{MSIS}$ \cite{ducas2017, kiltz2018}. The reduction $\mathcal{B}$ is unchanged, as it relays $H$ to the (quantum) $DS$ challenger and never programs it.
\end{remark}

\begin{theorem}\label{thm:withdraw}
The scheme of Algorithms \ref{alg:abs-keygen}--\ref{alg:abs-cverify} is withdrawable. For any computationally unbounded adversary $\adv$ making at most $q_W$ signing and $q_{H_M}$ random-oracle queries,
\[
\Adv^{\mathrm{Withdraw}}_{WS}(\adv):=\Pr\!\big[\Exp^{\mathrm{Withdraw}}_{WS,\adv}(1^\kappa)=1\big]\ \le\ \tfrac12+(q_W{+}1)\,\varepsilon_{\mathrm{zk}}+\frac{(q_W{+}1)(q_{H_M}{+}q_W{+}1)}{2^{\alpha}},
\]
where $q_W$ is the number of signing-oracle queries, $q_{H_M}$ the number of queries to $H_M$, and $\varepsilon_{\mathrm{zk}},\alpha$ are as in Lemma \ref{lem:zk} and Definition \ref{ass:entropy}. The bound is statistical, so it holds against computationally unbounded adversaries making at most $q_{H_M}$ random-oracle queries (assuming both key pairs are generated honestly by the experiment, as in Definition \ref{def:withdraw}). If instead $pk_v$ may be chosen by $\adv$, the bound gains an additive $\Adv^{\mathrm{MLWE}}$ term (Definition \ref{def:mlwe}), see Remark \ref{rem:wd}.
\end{theorem}

\begin{proof}
We bound $\Pr[\Exp^{\mathrm{Withdraw}}_{WS,\adv}(1^\kappa)=1]$ of Definition \ref{def:withdraw} (Algorithm \ref{alg:withdraw-exp}). The experiment generates both key pairs honestly, sets $\pi=\{pk_0,pk_1\}$, draws $b\sample\{0,1\}$, and gives $\adv$ the challenge $\sigma_b=\WSign(\mu^*,sk_b,\pi)$ together with the signing oracle $\Ora{WSign}_{sk_b,\pi}$ of Algorithm \ref{alg:withdraw-oracle}; finally $\adv$ outputs a guess $b'$ and wins iff $b'=b$. Let $\mathsf{Sim}(\mu,\pi)$ be the witness-free simulator of Lemma \ref{lem:zk}; crucially $\mathsf{Sim}$ depends only on the set $\pi$ (it simulates \emph{both} branches), not on which party is the prover.

The whole view of $\adv$ consists of $N:=q_W+1$ withdrawable signatures, all produced under the witness $sk_b$ (the challenge and the $q_W$ oracle replies), plus $\adv$'s $q_{H_M}$ queries to $H_M$.

\smallskip\noindent\textbf{Game $\mathsf{S}$ (witness-free).} Modify the experiment so that each of these $N$ signatures is produced by $\mathsf{Sim}(\cdot,\pi)$, programming the single induced $H_M$-entry (and aborting on a clash). Since $\mathsf{Sim}$ uses neither
$sk_0$ nor $sk_1$, the entire experiment $\mathsf{S}$ -- and hence the joint distribution of $\adv$'s view and output $b'$ -- is independent of $b$. Writing $p_{\mathsf S}:=\Pr[b'=1\mid\mathsf S]$, independence gives $\Pr[b'=b\mid\mathsf S]=\tfrac12\big((1-p_{\mathsf S})+p_{\mathsf S}\big)=\tfrac12$.

For $b\in\{0,1\}$ let $G_b$ be the real experiment conditioned on that bit, and $p_b:=\Pr[b'=1\mid G_b]$. We pass from $G_b$ to $\mathsf{S}$ by replacing the $N$ signatures one at a time. Each step turns one $\WSign(\cdot,sk_b,\pi)$ into one $\mathsf{Sim}(\cdot,\pi)$; by Lemma \ref{lem:zk}, applied with party $b$ as the real prover (valid because $pk_b$ is honest, so $sk_b$ is a genuine witness), the two are within statistical distance $\varepsilon_{\mathrm{zk}}$, provided the programmed point $(\mu,w_s,w_v,\pi)$ is fresh. By Definition \ref{ass:entropy} the commitment pair has min-entropy $\ge\alpha$, so at any step the freshness fails with probability at most $(q_{H_M}+N)/2^{\alpha}$ (the $H_M$-table holds at most $q_{H_M}+N$ points). Summing over the $N$ hybrids,
\[
\big|p_b-p_{\mathsf S}\big|\ \le\ \delta:=N\varepsilon_{\mathrm{zk}} +\frac{N(q_{H_M}+N)}{2^{\alpha}},\qquad b\in\{0,1\}.
\]

Using $\Pr[b'=0\mid G_0]=1-p_0$,
\[
\Pr[\Exp^{\mathrm{Withdraw}}_{WS,\adv}=1] =\tfrac12\big(\Pr[b'=0\mid G_0]+\Pr[b'=1\mid G_1]\big) =\tfrac12+\tfrac12\,(p_1-p_0).
\]
By the triangle inequality $|p_1-p_0|\le|p_1-p_{\mathsf S}|+|p_{\mathsf S}-p_0|\le 2\delta$, hence
\[
\Pr[\Exp^{\mathrm{Withdraw}}_{WS,\adv}=1]\ \le\ \tfrac12+\delta =\tfrac12+(q_W+1)\varepsilon_{\mathrm{zk}}+\frac{(q_W+1)(q_{H_M}+q_W+1)}{2^{\alpha}} .
\]
Every step is information-theoretic, so the bound holds against computationally unbounded adversaries that make at most $q_{H_M}$ random-oracle queries.
\end{proof}

\begin{remark}[Theorem~\ref{thm:withdraw} in the QROM]\label{rem:qrom-wd}
Theorem~\ref{thm:withdraw} is stated in the classical ROM. Its only use of the random oracle is the witness-free simulation of $\WSign$, which programs $H_M$ at the adaptively chosen, high- min-entropy point $(\mu,w_s,w_v,\pi)$ (Assumption~\ref{ass:entropy}). In the QROM the adaptive-reprogramming lemma of Grilo--H\"ovelmanns--H\"ulsing--Majenz \cite{grilo2021} applies exactly as in Remark~\ref{rem:qrom}, replacing the classical freshness term $\tfrac{(q_W{+}1)(q_{H_M}{+}q_W{+}1)}{2^{\alpha}}$ by $\tfrac{3}{2}(q_W{+}1)\sqrt{(q_{H_M}{+}q_W{+}1)/2^{\alpha}}$; the statistical, information-theoretic core of the argument is unchanged. For the no-hint instantiation the two real signing distributions are identical (Lemma~\ref{lem:perfectzk}), so the $dk_v$/$sk_v$ bound of Proposition~\ref{prop:sdvs} is exact in both models and incurs no reprogramming term.
\end{remark}

\begin{remark}[Adversarial $pk_v$; strong-DV variant]\label{rem:wd} Definition \ref{def:withdraw} generates both keys honestly, so the proof is statistical. In the stronger model where $\adv$ chooses $pk_v$, $\mathsf{Sim}$ still simulates for any $pk_v$, but the two real provers require both statements to admit short witnesses; bridging an honestly generated $pk_v$ to an adversarial one costs one decisional-MLWE step (Definition \ref{def:mlwe}), adding $\Adv^{\mathrm{MLWE}}$ to the bound. For the strong designated-verifier scheme of Section \ref{sec:sdvs}, the same hybrid first replaces every signature by a simulated one and then replaces the two ciphertexts and the one-time pad by uniform strings; this yields $\Pr[\Exp^{\mathrm{Withdraw}}=1]\le\tfrac12+(q_W{+}1)\varepsilon_{\mathrm{zk}} +q_W\big(2\,\Adv^{\mathrm{pr}}_{\PKE}+\Adv^{\mathrm{prg}}_{G}\big)+\tfrac{(q_W{+}1)(q_{H_M}{+}q_W{+}1)}{2^{\alpha}}$, i.e.\ withdrawability against third parties as in Proposition \ref{prop:sdvs}, now made quantitative; non-transferability against the designated verifier is exactly the statistical statement proved above.
\end{remark}

\subsection{An instantiation}\label{subsec:instantiation}
The general construction of Section \ref{sec:abstract} for a withdrawable digital signature scheme whose underlying scheme is built using the Fiat--Shamir with aborts paradigm was proven to be unforgeable under insider corruption and withdrawable. As a direct application of this proposal, in this section we present a withdrawable signature scheme based on a no-hint, full-$\bt$ Dilithium-style Fiat--Shamir-with-aborts signature~\cite{ducas2017}; adapting to production ML-DSA (with hints, and a nonzero $\varepsilon_{\mathrm{zk}}$) is discussed after Corollary~\ref{cor:concrete}. The verification algorithm below is \emph{public}; it serves as a building block ($\WSVerify^{\mathsf{pub}}$) that Section \ref{sec:sdvs} encrypts to the designated verifier to obtain the strong designated-verifier withdrawable signature of Definition \ref{def:ws}.

\begin{algorithm}
\caption{Dilithium signature \cite{ducas2017}}\label{alg:dilithium}
\begin{algorithmic}[1]
\Procedure{KeyGen}{$1^\kappa$}
  \State $\bA\sample\Rq^{k\times l},\ (\bs_1,\bs_2)\sample S_\eta^l\times S_\eta^k$
  \State $\bt=\bA\bs_1+\bs_2$
  \State \textbf{return} $(pk=(\bA,\bt),\ sk=(\bA,\bt,\bs_1,\bs_2))$
\EndProcedure
\Procedure{Sign}{$\mu,sk$}
  \Repeat
    \State $\by\sample S_{\gamma_1-1}^l$
    \State $\bw=\HighBits(\bA\by,2\gamma_2)$
    \State $c\in\mathcal{C}=H(\mu,\bw)$
    \State $\bz=\by+c\bs_1$
  \Until{$(\norminf{\bz}<\gamma_1-\beta)\wedge(\norminf{\LowBits(\bA\by-c\bs_2,2\gamma_2)}<\gamma_2-\beta)$}
  \State \textbf{return} $\sigma=(\bz,c)$
\EndProcedure
\Procedure{Verify}{$\mu,pk,\sigma$}
  \State $\bw'=\HighBits(\bA\bz-c\bt,2\gamma_2)$
  \If{$\norminf{\bz}<\gamma_1-\beta$ and $c=H(\mu,\bw')$} \textbf{return} $1$
  \EndIf
\EndProcedure
\end{algorithmic}
\end{algorithm}

Taking the above algorithm as starting point and following the general construction in Section \ref{sec:abstract}, it is straightforward to set a withdrawable digital signature based on Dilithium.

Since Dilithium relies on the hardness of MSIS, Theorem \ref{thm:unforge} applies to prove that the proposal below is unforgeable under insider corruption. Theorem \ref{thm:withdraw} then proves withdrawability of the honest-key public core \emph{statistically}, under Assumption \ref{ass:entropy} and perfect naHVZK; the hardness of the decisional MLWE problem is used only for adversarially chosen verifier keys (Remark \ref{rem:wd}) and for the privacy of the encryption layer in Section \ref{sec:sdvs}.

Below follows the withdrawable signature based on Dilithium:

\begin{algorithm}
\caption{$\KeyGen(pp)$}\label{alg:dil-keygen}
\begin{algorithmic}[1]
\State parse $\bA$ from $pp$
\State $(\bs_1,\bs_2)\sample S_\eta^l\times S_\eta^k$
\State $\bt=\bA\bs_1+\bs_2$
\State $(ek,dk)\leftarrow\mathsf{KGen}_e(1^\kappa)$
\State \textbf{return} $pk=(\bA,\bt,ek),\ sk=(pk,\bs_1,\bs_2,dk)$
\end{algorithmic}
\end{algorithm}

\begin{algorithm}
\caption{$\WSign^{\mathsf{pub}}(\mu,sk_s,\pi)$}\label{alg:dil-wsign}
\begin{algorithmic}[1]
\State $\pi=(pk_s,pk_v)$
\Repeat \Comment{simulate verifier branch from the accepted distribution}
\State $c_v\sample\mathcal{C};\ \bz_v\sample S_{\gamma_1-\beta-1}^l;\ w_v=\HighBits(\bA\bz_v-c_v\bt_v,2\gamma_2)$
\Until{$\norminf{\LowBits(\bA\bz_v-c_v\bt_v,2\gamma_2)}<\gamma_2-\beta$}
\Repeat
  \State $\by\sample S_{\gamma_1-1}^l;\ w_s=\HighBits(\bA\by,2\gamma_2)$
  \State $g=H_M(\mu,w_s,w_v,\pi);\ c_s=\iota\big((g-\iota^{-1}(c_v))\bmod M\big)$
  \State $\bz_s=\by+c_s\bs_1'$
\Until{$(\norminf{\bz_s}<\gamma_1-\beta)\wedge(\norminf{\LowBits(\bA\by-c_s\bs_2',2\gamma_2)}<\gamma_2-\beta)$}
\State \textbf{return} $\sigma=(c_s,c_v,\bz_s,\bz_v)$
\end{algorithmic}
\end{algorithm}

Throughout, the signer instance of a key pair is written $sk_s=(pk_s,\bs_1',\bs_2',dk_s)$ with $pk_s=(\bA,\bt_s,ek_s)$, and the verifier instance $sk_v=(pk_v,\bs_1'',\bs_2'',dk_v)$; the primes are the per-party renaming of the generic $(\bs_1,\bs_2)$ output by $\KeyGen(pp)$.

\begin{remark}\label{rem:sym}
$\WSign^{\mathsf{pub}}$ is invoked by \emph{either} party: the caller runs its own branch for real (with the rejection loop) and simulates the other from the accepted distribution above. A verifier signs by swapping the roles of $s$ and $v$. Because both branches are now drawn from the accepted low-order distribution, the two invocations produce identically distributed $\sigma$ (Lemma~\ref{lem:perfectzk}, $ \varepsilon_{\mathrm{zk}}=0$), and the public predicates $\norminf{\LowBits(\bA\bz_s-c_s\bt_s,2\gamma_2)}<\gamma_2-\beta$ and $\norminf{\LowBits(\bA\bz_v-c_v\bt_v,2\gamma_2)}<\gamma_2-\beta$ hold simultaneously, revealing nothing about the origin.
\end{remark}

\begin{algorithm}
\caption{$\WSVerify^{\mathsf{pub}}(\mu,\pi,\sigma)$}\label{alg:dil-wsverify}
\begin{algorithmic}[1]
\State $\sigma=(c_s,c_v,\bz_s,\bz_v)$
\State $w_s'=\HighBits(\bA\bz_s-c_s\bt_s,2\gamma_2);\ w_v'=\HighBits(\bA\bz_v-c_v\bt_v,2\gamma_2)$
\State $g=H_M(\mu,w_s',w_v',\pi)$
\If{$(\iota^{-1}(c_s)+\iota^{-1}(c_v)\equiv g \!\!\pmod{M})\wedge(\norminf{\bz_s}<\gamma_1-\beta)\wedge(\norminf{\bz_v}<\gamma_1-\beta)$} \textbf{return} $1$
\EndIf
\end{algorithmic}
\end{algorithm}

\begin{algorithm}
\caption{$\Confirm^{\mathsf{pub}}(\mu,sk_s,\pi,\sigma)$}\label{alg:dil-confirm}
\begin{algorithmic}[1]
\State \textbf{return} $\tilde{\sigma}\leftarrow\Sign_{sk_s}(\mu\,\|\,\pi\,\|\,\sigma)$ \Comment{ordinary Dilithium signature, Alg. \ref{alg:dilithium}}
\end{algorithmic}
\end{algorithm}

\begin{algorithm}
\caption{$\CVerify^{\mathsf{pub}}(\mu,\pi,\sigma,\tilde{\sigma})$}\label{alg:dil-cverify}
\begin{algorithmic}[1]
\State $\sigma=(c_s,c_v,\bz_s,\bz_v)$
\If{$\WSVerify^{\mathsf{pub}}(\mu,\pi,\sigma)=1\ \wedge\ \Verify_{pk_s}(\mu\,\|\,\pi\,\|\,\sigma,\ \tilde{\sigma})=1$} \textbf{return} $1$
\EndIf
\end{algorithmic}
\end{algorithm}

\begin{remark}[Correctness]\label{rem:rr}
For the real branch, the rejection condition $\norminf{\LowBits(\bA\by-c_s\bs_2',2\gamma_2)}<\gamma_2-\beta$ yields $\HighBits(\bA\bz_s-c_s\bt_s,2\gamma_2)=\HighBits(\bA\by,2\gamma_2)=w_s$; the simulated branch matches by construction; and $\iota^{-1}(c_s)+\iota^{-1}(c_v)\equiv g$ by definition of $c_s$. Hence $\WSVerify$ accepts honest signatures, and $\CVerify$ accepts by correctness of $\Sign$. The quantity $r$ of the original construction is no longer used, so the former $r=r'$ identity is unnecessary.
\end{remark}

\paragraph{Instantiation of $\varepsilon_{\mathrm{zk}}$ and $\alpha$.}
Our scheme uses the no-hint identification scheme: verification recomputes $w'=\HighBits(\bA\bz-c\bt,2\gamma_2)$ from the full $\bt$ (no $\bt_0/\bt_1$ split, no hint). For this scheme the simulator of Lemma \ref{lem:zk} is perfect.

\begin{lemma}[Perfect naHVZK]\label{lem:perfectzk}
Conditioned on the programmed $H_M$-point being fresh, the simulator $\mathsf{Sim}(\mu,\pi)$ of Lemma \ref{lem:zk} produces a transcript identically distributed to a genuine $\WSign(\mu,sk_s,\pi)$; that is, $\varepsilon_{\mathrm{zk}}=0$.
\end{lemma}

\begin{proof}
Fix the challenge $c$ (uniform in the real game as $H_M$ is a random oracle, and uniform in the simulation by programming a fresh point). In $\WSign$, $\by$ is uniform over $S_{\gamma_1-1}^l$ and $\bz=\by+c\bs_1'$; conditioned on $\norminf{\bz}<\gamma_1-\beta$, $\bz$ is uniform over $S_{\gamma_1-\beta-1}^l$, because each target $\bz$ in that box has a unique in-range preimage $\by=\bz-c\bs_1'$ (as $\norminf{c\bs_1'}\le\beta$). The further test $\norminf{\LowBits(\bA\by-c\bs_2',2\gamma_2)}<\gamma_2-\beta$ equals $\norminf{\LowBits(\bA\bz-c\bt_s,2\gamma_2)}<\gamma_2-\beta$ (since $\bA\bz-c\bt_s=\bA\by-c\bs_2'$), a function of $(c,\bz)$ and the public $\bt_s$ only; and the recomputed commitment is $\HighBits(\bA\bz-c\bt_s,2\gamma_2)=w_s$. The corrected simulator (Lemma \ref{lem:zk}) draws $\bz$ uniformly over $S_{\gamma_1-\beta-1}^l$ with the identical low-order test and sets $w_s=\HighBits(\bA\bz-c\bt_s,2\gamma_2)$. Hence the conditional law of $(c,\bz,w_s)$ is identical in both, and the verifier branch is produced by the same simulation on both sides. The transcripts therefore coincide.
\end{proof}

\begin{table}[h]\centering\small
\begin{tabular}{lccccc}
\toprule
 & $(k,l)$ & $\eta,\tau,\beta$ & $\gamma_1,\gamma_2$ & $\tfrac{q-1}{2\gamma_2}$ & $\alpha\ (\ge)$ \\
\midrule
Dilithium2 / ML-DSA-44 & $(4,4)$ & $2,39,78$ & $2^{17},\,95232$ & $44$ & $1397$ \\
Dilithium3 / ML-DSA-65 & $(6,5)$ & $4,49,196$ & $2^{19},\,261888$ & $16$ & $1024$ \\
Dilithium5 / ML-DSA-87 & $(8,7)$ & $2,60,120$ & $2^{19},\,261888$ & $16$ & $1024$ \\
\bottomrule
\end{tabular}
\caption{Parameters ($n=256$, $q=8380417$) and the resulting commitment min-entropy
lower bound $\alpha$ (bits). Challenge entropies $\log_2|B_\tau|=192.8,225.3,257.0$.}
\end{table}

\begin{corollary}[Concrete security]\label{cor:concrete}
Instantiated with no-hint Dilithium, the scheme has $\varepsilon_{\mathrm{zk}}=0$ (Lemma \ref{lem:perfectzk}) and $\alpha\ge 1024$ (Assumption~\ref{ass:entropy}). Hence, for any adversary making $q_W$ signing and $q_{H_M}$ random-oracle queries, Theorems \ref{thm:unforge} and \ref{thm:withdraw} give
\begin{align*}
\Adv^{\eufic}_{WS}(\adv)&\le \Adv^{\eufcma}_{DS}(\mathcal{B})+q_W(q_{H_M}+q_W)\,2^{-1024},\\
\Adv^{\mathrm{Withdraw}}_{WS}(\adv)&\le \tfrac12+(q_W{+}1)(q_{H_M}{+}q_W{+}1)\,2^{-1024}.
\end{align*}
Unforgeability reduces to the EUF-CMA security of the underlying no-hint Dilithium-style signature (up to the $2^{-1024}$ programming term), whereas withdrawability follows \emph{ statistically} from perfect naHVZK ($\varepsilon_{\mathrm{zk}}=0$) and Assumption~\ref{ass:entropy}, up to the same $2^{-1024}$ programming term; it does not reduce to EUF-CMA.
\end{corollary}

\begin{remark}[Hint-optimised Dilithium]
Production Dilithium publishes $\bt_1$ (not $\bt$) plus a hint $\mathbf{h}$, which trades the exact recovery $\HighBits(\bA\bz-c\bt)=w_s$ for $\mathsf{UseHint}$. That variant has a small, standard naHVZK defect in place of $\varepsilon_{\mathrm{zk}}=0$ and slightly smaller signatures; our no-hint choice keeps the proofs clean (perfect naHVZK) at a modest size cost. Either way the asymptotic and concrete security are governed by the same $\Adv^{\eufcma}_{DS}$ and $2^{-\alpha}$ terms.
\end{remark}

\subsection{Strong designated-verifier variant}\label{sec:sdvs}

The scheme of Sections \ref{sec:abstract} and 3.2 is publicly verifiable. To match the Liu--Baek--Susilo model \cite{liu2023}, where a withdrawable signature is verifiable \emph{only} by the designated verifier, we upgrade it to a strong designated-verifier (SDVS) scheme: we hybrid-encrypt the publicly-verifiable signature to the verifier (and to the signer, so that $\Confirm$ stays stateless). Write $(\WSign^{\mathsf{pub}},\WSVerify^{\mathsf{pub}},\Confirm^{\mathsf{pub}},\CVerify^{\mathsf{pub}})$ for the scheme of Section \ref{subsec:instantiation} (Algorithms \ref{alg:dil-wsign}--\ref{alg:dil-cverify}), where $\WSVerify^{\mathsf{pub}}(\mu,\pi,\rho)$ is public; $\KeyGen$ is as amended (each party holds $(ek,dk)$). Let $\rho$ denote a public ring signature and $\|\rho\|$ its bit length.

\begin{algorithm}
\caption{$\WSign(\mu,sk_s,\pi)$ (strong DVS)}\label{alg:sdvs-wsign}
\begin{algorithmic}[1]
\State $\rho\leftarrow\WSign^{\mathsf{pub}}(\mu,sk_s,\pi)$ \Comment{public OR-proof ring signature}
\State $K\sample\{0,1\}^\kappa$
\State $ct_s\leftarrow\mathsf{Enc}(ek_s,K);\quad ct_v\leftarrow\mathsf{Enc}(ek_v,K)$
\State $C\leftarrow\rho\oplus G(K)$ \Comment{$G(K)$ truncated to $\|\rho\|$ bits}
\State \textbf{return} $\sigma=(ct_s,ct_v,C)$
\end{algorithmic}
\end{algorithm}

\begin{algorithm}
\caption{$\WSVerify(\mu,sk_v,pk_s,\sigma)$ (strong DVS; designated)}\label{alg:sdvs-wsverify}
\begin{algorithmic}[1]
\State parse $pk_v$ from $sk_v$;\ $\pi=\{pk_s,pk_v\}$; parse $\sigma=(ct_s,ct_v,C)$
\State $K\leftarrow\mathsf{Dec}(dk_v,ct_v);\quad \rho\leftarrow C\oplus G(K)$
\State \textbf{return} $\WSVerify^{\mathsf{pub}}(\mu,\pi,\rho)$
\end{algorithmic}
\end{algorithm}

\begin{algorithm}
\caption{$\Confirm(\mu,sk_s,\pi,\sigma)$ (strong DVS)}\label{alg:sdvs-confirm}
\begin{algorithmic}[1]
\State parse $\sigma=(ct_s,ct_v,C)$
\State $K\leftarrow\mathsf{Dec}(dk_s,ct_s);\quad \rho\leftarrow C\oplus G(K)$ \Comment{signer recovers $\rho$ via its own $dk_s$}
\State $\bar{\sigma}\leftarrow\Sign_{sk_s}(\mu\,\|\,\pi\,\|\,\sigma)$
\State \textbf{return} $\tilde{\sigma}=(\bar{\sigma},\rho,K)$
\end{algorithmic}
\end{algorithm}

\begin{algorithm}
\caption{$\CVerify(\mu,\pi,\sigma,\tilde{\sigma})$ (strong DVS; public)}\label{alg:sdvs-cverify}
\begin{algorithmic}[1]
\State parse $\sigma=(ct_s,ct_v,C),\ \tilde{\sigma}=(\bar{\sigma},\rho,K)$
\If{$\big(C=\rho\oplus G(K)\big)\wedge\big(\WSVerify^{\mathsf{pub}}(\mu,\pi,\rho)=1\big)\wedge\big(\Verify_{pk_s}(\mu\,\|\,\pi\,\|\,\sigma,\bar{\sigma})=1\big)$} \textbf{return} $1$
\EndIf
\end{algorithmic}
\end{algorithm}

\begin{remark}[Scope of confirmed verification]\label{rem:cverify-scope}
$\CVerify$ certifies that the signer authenticated the outer object $\sigma$ (via $\bar\sigma=\Sign_{sk_s}(\mu\|\pi\|\sigma)$) and that the revealed pad opens $C$ to a public core $\rho$ that passes $\WSVerify^{\mathsf{pub}}$. It does \emph{not} publicly certify ciphertext consistency, i.e.\ that $ct_s$ and $ct_v$ encrypt the same $K$: for randomised $\PKE$ this cannot be checked without opening the encryption randomness. Confirmation is therefore a signer-authenticated opening of $\sigma$ to a valid ambiguous core, which is all correctness (Definition~\ref{def:correct}) and unforgeability (Theorem~\ref{thm:unforge}) require; it does not retroactively attest that the designated verifier could decrypt $\sigma$ prior to confirmation. For honestly generated $\sigma$ the two ciphertexts do share $K$, so $\WSVerify=1\Rightarrow\CVerify=1$ as stated.
\end{remark}

\begin{proposition}[SDVS variant]\label{prop:sdvs}
Assume $\PKE$ has ciphertext pseudorandomness (advantage $\Adv^{\mathrm{pr}}_{\PKE}$) and $G$ is a secure PRG. Then Algorithms \ref{alg:sdvs-wsign}--\ref{alg:sdvs-cverify} form a correct withdrawable signature (Definitions~\ref{def:ws}--\ref{def:withdraw}) with \emph{designated} verification, and:
\begin{itemize}
  \item \emph{(Designated verification / privacy.)} Without $dk_s$ or $dk_v$, $\sigma$ is pseudorandom; in particular third parties cannot publicly verify $\sigma$; recovering $\rho$ (and hence verifying) requires a decryption key, so both the designated verifier (via $dk_v$) and the signer (via $dk_s$) can verify.
  \item \emph{(Withdrawability.)} For any PPT third-party adversary making $q_W$ signing and $q_{H_M}$ random-oracle queries,
  $\Adv^{\mathrm{Withdraw}}:=\Pr[\Exp^{\mathrm{Withdraw}}_{WS,\adv}=1]\le \tfrac12+(q_W{+}1)\big(2\,\Adv^{\mathrm{pr}}_{\PKE}+\Adv^{\mathrm{prg}}_{G}\big)$, where $\Adv^{\mathrm{pr}}_{\PKE}$ is the ciphertext-pseudorandomness advantage of $\PKE$. Against an adversary holding the \emph{full} verifier secret key $sk_v=(pk_v,\bs_1'',\bs_2'',dk_v)$ the bound is $\tfrac12+(q_W{+}1)\,\varepsilon_{\mathrm{zk}}$: by the perfect naHVZK of Lemma~\ref{lem:perfectzk} the verifier branch is identically distributed whether produced for real or simulated, so the witness confers no distinguishing power. This equals exactly $\tfrac12$ for the no-hint Dilithium instantiation.
  \item \emph{(Unforgeability under insider corruption.)} Reduces to the EUF-CMA security of $\Sign$ exactly as in Theorem \ref{thm:unforge}; the reduction chooses $K$ and the ciphertexts itself, hence needs no secret key.
\end{itemize}
\end{proposition}

\begin{proof}
The strong-DV algorithms (Algorithms \ref{alg:sdvs-wsign}--\ref{alg:sdvs-cverify}) are built on the public core 
\[
(\WSign^{\mathsf{pub}},\WSVerify^{\mathsf{pub}},\Confirm^{\mathsf{pub}}, \CVerify^{\mathsf{pub}})
\] 
of Section \ref{subsec:instantiation}, a public-key encryption scheme $\PKE=(\mathsf{KGen}_e,\Enc,\Dec)$ with decryption-failure probability $\delta_{\PKE}$ and ciphertext-pseudorandomness advantage $\Adv^{\mathrm{pr}}_{\PKE}$ (i.e.\ for any message $m$, $\Enc(ek,m)$ is indistinguishable from a uniform ciphertext given $ek$ but not $dk$; this implies IND-CPA and key-privacy), and a PRG $G$ with advantage $\Adv^{\mathrm{prg}}_{G}$. Recall $\sigma=(ct_s,ct_v,C)$, $ct_s=\Enc(ek_s,K)$, $ct_v=\Enc(ek_v,K)$, $C=\rho\oplus G(K)$, $\rho=\WSign^{\mathsf{pub}}(\mu,sk_s,\pi)$, $K\xleftarrow{\$}\{0,1\}^\kappa$.

The verifier recovers $K=\Dec(dk_v,ct_v)$ and $\rho=C\oplus G(K)$, so 
\[
\WSVerify=\WSVerify^{\mathsf{pub}}(\mu,\pi,\rho)=1
\] 
by correctness of the core. $\Confirm$ recovers the same $(\rho,K)$ via $dk_s$ and outputs $\tilde\sigma=(\bar\sigma,\rho,K)$, $\bar\sigma=\Sign_{sk_s}(\mu\|\pi\|\sigma)$; then $\CVerify$ makes the following verifications: $C=\rho\oplus G(K)$ (true), $\WSVerify^{\mathsf{pub}}(\mu,\pi,\rho)=1$, and $\Verify_{pk_s}(\mu\|\pi\|\sigma,\bar\sigma)=1$, all of which hold. The error is at most $2\delta_{\PKE}$ plus the negligible correctness errors of the core and of $\Sign$.

We show $\sigma$ is pseudorandom to any party lacking both $dk_s,dk_v$. Hybrids: $\mathsf{H}_0=$ real $\sigma$; $\mathsf{H}_1$ replaces $ct_s$ by a uniform string; $\mathsf{H}_2$ also replaces $ct_v$ by a uniform string; $\mathsf{H}_3$ also replaces $G(K)$ by a uniform string $R$. Each PKE step costs $\le\Adv^{\mathrm{pr}}_{\PKE}$ (the reduction picks $K$, embeds the challenge as $ct_s$ resp.\ $ct_v$, and computes the rest from the known $K$; no decryption key is used); the PRG step costs $\le\Adv^{\mathrm{prg}}_{G}$ (in $\mathsf{H}_2$, $K$ feeds only $G$). In $\mathsf{H}_3$, $C=\rho\oplus R$ is uniform and independent of $\rho$, so $\sigma$ is uniform. Hence $\Delta(\sigma,\mathcal U)\le 2\Adv^{\mathrm{pr}}_{\PKE}+\Adv^{\mathrm{prg}}_{G}$. Consequently no public procedure can decide validity of $\sigma$ (it would distinguish $\mathsf{H}_0$ from $\mathsf{H}_3$); verification requires a decryption key, so the designated verifier (via $dk_v$) and the signer (via $dk_s$) can verify, while third parties cannot.

In Definition \ref{def:withdraw} the adversary $\adv$ holds no secret keys and sees $N:=q_W+1$ signatures (challenge and oracle replies) under the bit $b$. Applying the previous hybrid to each, every signature is within $2\Adv^{\mathrm{pr}}_{\PKE}+\Adv^{\mathrm{prg}}_{G}$ of uniform \emph{independently of $b$} (a uniform $\sigma$ reveals neither $\rho$ nor which keys were used). Thus each of the two conditional views is within $N(2\Adv^{\mathrm{pr}}_{\PKE}+\Adv^{\mathrm{prg}}_{G})$ of the same $b$-independent all-uniform view, whence
\[
\Pr[\Exp^{\mathrm{Withdraw}}_{WS,\adv}=1]\le\tfrac12+(q_W{+}1)\big(2\Adv^{\mathrm{pr}}_{\PKE}+\Adv^{\mathrm{prg}}_{G}\big).
\]

If $\adv$ holds the full verifier secret key $sk_v=(pk_v,\bs_1'',\bs_2'',dk_v)$ (in particular its $dk_v$), it decrypts and recovers $\rho=\WSign^{\mathsf{pub}}(\mu,sk_b,\pi)$; the ciphertexts and $C$ are then determined by $(\rho,K)$ and carry no further information about $b$. Here we do \emph{not} route through the simulator/programming hybrid of Theorem~\ref{thm:withdraw}; instead we compare the two \emph{real} public-core distributions $\WSign^{\mathsf{pub}}(\cdot,sk_0,\pi)$ and $\WSign^{\mathsf{pub}}(\cdot,sk_1,\pi)$ directly. In each, one branch is produced for real and the other is drawn from the accepted low-order distribution, and in both cases the challenge $g=H_M(\mu,w_s,w_v,\pi)$ is read from the honest random oracle. By the per-branch naHVZK of Lemma~\ref{lem:zk}, a real accepted branch and a simulated accepted branch are within statistical distance $\varepsilon_{\mathrm{zk}}$ \emph{as distributions}, with no oracle reprogramming; hence the two real signing distributions are within $2\varepsilon_{\mathrm{zk}}$ per signature, and no $2^{-\alpha}$ freshness term is incurred. Summing over the $N=q_W{+}1$ signatures, $\Pr[\Exp^{\mathrm{Withdraw}}=1]\le\tfrac12+(q_W{+}1)\varepsilon_{\mathrm{zk}}$. For the no-hint Dilithium instantiation $\varepsilon_{\mathrm{zk}}=0$ (Lemma~\ref{lem:perfectzk}), so the two real distributions are \emph{identical} and the excess advantage $\Pr[\Exp^{\mathrm{Withdraw}}=1]-\tfrac12$ is exactly $0$: even the designated verifier cannot tell whether it or the signer produced $\sigma$.

The reduction is that of Theorem \ref{thm:unforge}, with the reduction performing the encryption itself. $\mathcal B$ receives the $\Sign$-challenge key $pk^\ast$, sets $pk_s:=pk^\ast$, $\mathcal B$ receives the $\Sign$-challenge key $pk^\ast$ (the signing component only). It generates a fresh encryption key pair $(ek_s,dk_s)\leftarrow\mathsf{KGen}_e(1^\kappa)$ itself and sets $pk_s:=(pk^\ast,ek_s)$, keeping $dk_s$ private; it never uses $dk_s$ (confirmation replies are served from stored $(\rho,K)$, and signing queries encrypt under the public $ek_s$). It generates all other keys honestly (so it holds $sk_v$ and may give it to $\adv$). It answers: $\Ora{WSign}(\mu)$ by simulating $\rho$ with the witness-free simulator of Lemma \ref{lem:zk} (programming $H_M$), drawing $K$, setting $ct_s=\Enc(ek_s,K)$, $ct_v=\Enc(ek_v,K)$, $C=\rho\oplus G(K)$, storing $(\sigma,\rho,K)$ in $W$, and returning $\sigma$; and $\Ora{Confirm}(\mu,\sigma)$, for $\sigma\in W$, by retrieving the stored $(\rho,K)$ and returning $(\Sign_{sk^\ast}(\mu\|\pi\|\sigma),\rho,K)$ through the signing oracle. Thus $\mathcal B$ uses neither $sk_s$ nor any decryption key. A confirmed forgery $(\mu^\ast,\sigma^\ast,\tilde\sigma^\ast=(\bar\sigma^\ast,\rho^\ast,K^\ast))$ with $\mu^\ast$ never confirmed satisfies $\Verify_{pk^\ast}(\mu^\ast\|\pi\|\sigma^\ast,\bar\sigma^\ast)=1$ on the fresh message $\mu^\ast\|\pi\|\sigma^\ast$, which $\mathcal B$ outputs. As in Theorem \ref{thm:unforge},
\[
\Adv^{\eufic}_{WS}(\adv)\le \Adv^{\eufcma}_{DS}(\mathcal B)+q_W\,\varepsilon_{\mathrm{zk}}+q_W(q_{H_M}+q_W)/2^{\alpha},
\]
with no contribution from $\PKE$ or $G$ (which $\mathcal B$ evaluates honestly).
\end{proof}

\subsection{Comparison with Liu--Baek--Susilo}\label{sec:comparison}

Our construction realises the same primitive and the same security model as Liu, Baek and Susilo \cite{liu2023}: the syntax of Definition \ref{def:ws} and the security notions of Definitions \ref{def:correct}, \ref{def:euf} and \ref{def:withdraw} are theirs, and we follow their two-stage ``withdraw, then confirm'' template. The differences are confined to the \emph{construction}, and each of them follows from a single structural fact about lattices.

\paragraph{Root cause: the group shared secret has no cheap lattice analogue.}
The designation mechanism of \cite{liu2023} is, at its core, a Diffie--Hellman (or pairing) shared secret embedded in the signature. In their Schnorr-based scheme the withdrawable signature contains $\sigma_2=pk_v^{\,z-rt}$, which the designated verifier checks by raising to its secret key $sk_v$ via $(g^a)^b=(g^b)^a$; the BLS-based scheme uses $e(\cdot,g^{sk_v})$ in the same way. A single algebraic object thereby delivers \emph{both} properties required of a withdrawable signature at once: signer/verifier ambiguity --- either party can compute $\sigma$ from its own secret --- and designation --- verification requires $sk_v$ --- using nothing but the parties' ordinary public keys, in effect a non-interactive key agreement (NIKE). This shared secret is essentially free in a group but has no cheap lattice counterpart: the noisy analogue of $(g^a)^b\approx(g^b)^a$ agrees only approximately, so a literal transcription that blinds the signature by a uniform matrix $\bB\in\Rq^{k\times k}$ leaves a residue $\bB\bs_2''$ that the verifier cannot cancel (it is given $\bB\bA$, never $\bB$). Lattice NIKE is in fact possible: after a long line of negative evidence, including information-theoretic efficiency barriers for polynomial-modulus reconciliation \cite{guo2022}, the recent SWOOSH scheme \cite{gajland2024} gives the first practical M-LWE NIKE --- but at a cost (public keys of hundreds of kilobytes, a super-polynomial modulus-to-noise ratio, and a NIZK for active security) well above what designation needs. Crucially, designation does \emph{not} require a shared secret at all: it only needs to deliver a key to the designated verifier, which is public-key encryption --- a weaker and far cheaper primitive. We therefore keep ambiguity and designation as \emph{separate} mechanisms, an OR-proof for the former and encryption for the latter; this is the origin of every row of Table \ref{tab:comparison}.

\begin{table*}[t]
\centering\footnotesize
\renewcommand{\arraystretch}{1.25}
\begin{tabularx}{\textwidth}{@{}>{\raggedright\arraybackslash}p{2.4cm} X X >{\raggedright\arraybackslash}p{2.4cm}@{}}
\toprule
\textbf{Aspect} & \textbf{Liu--Baek--Susilo \cite{liu2023}} & \textbf{This work} & \textbf{Reason} \\
\midrule
Designation of $\sigma$ & $sk_v$ enters the verification exponent (Diffie--Hellman / pairing) & encrypt $\sigma$ to $pk_v$ via a PKE & no free shared secret \\
Signer/verifier ambiguity & algebraic symmetry of $\sigma$ & two-party FS-with-aborts OR-proof (witness-indist.) & no DH symmetry \\
Confirmation & re-derive shared randomness $r$; rebuild a linked $\tilde\sigma$ & ordinary signature $\Sign_{sk_s}(\mu\,\|\,\pi\,\|\,\sigma)$ & FS randomness not reusable \\
Key material & one key pair (signing $=$ designation) & signing key pair $+$ encryption key pair $(ek,dk)$ & signing key cannot designate \\
Withdrawability from & DDH (Schnorr) / DBDH (BLS) & decisional MLWE (PKE) $+$ statistical WI & lattice counterpart \\
Unforgeability from & CDH / DL & MSIS / SelfTargetMSIS & lattice counterpart \\
Overall structure & monolithic designated-verifier object & public ambiguous core $+$ encryption layer & composition forced \\
\bottomrule
\end{tabularx}
\caption{Comparison of the Liu--Baek--Susilo withdrawable signature with the lattice construction of this work. Every difference follows from the absence of a free Diffie--Hellman-style shared secret in the lattice setting (Section \ref{sec:comparison}).}
\label{tab:comparison}
\end{table*}

\paragraph{Ambiguity and confirmation.}
For ambiguity we use a two-clause Fiat--Shamir-with-aborts OR-proof (Algorithm \ref{alg:dil-wsign}): a proof of knowledge of a short witness for $\bt_s$ \emph{or} for $\bt_v$, which either party can produce by executing its own branch and simulating the other. This is the natural lattice substitute for the group symmetry, and it agrees with the authors' own description of withdrawable signatures as two-party ring signatures augmented with a transformation stage \cite{liu2023}. For confirmation, where \cite{liu2023} reconstructs the shared randomness $r$ and rebuilds an algebraically linked $\tilde\sigma$, we let $\Confirm$ emit an ordinary signature $\Sign_{sk_s}(\mu\,\|\,\pi\,\|\,\sigma)$ on the entire withdrawable object: the masking randomness of Fiat--Shamir with aborts is not reusable in the way a discrete-logarithm nonce is, and binding by a signature is both simpler and reducible to the unforgeability of the underlying scheme (Theorem \ref{thm:unforge}).

\paragraph{Designation by encryption: a PKE, not a KEM.}
Lacking the free shared secret, we transport the designation explicitly: the scheme of Section \ref{sec:sdvs} encrypts the publicly verifiable, ambiguous signature $\rho$ to the verifier, so that only a holder of $dk_v$ can recover and check it. This encryption must place \emph{the same} chosen key $K$ under \emph{two} public keys: $ek_v$, so that the verifier can decrypt and verify, and $ek_s$, so that the signer can decrypt and run $\Confirm$ statelessly (it is given only $\sigma$). A key-encapsulation mechanism returns a fresh \emph{random} key for each encapsulation and recipient, so $\mathsf{Encaps}(ek_v)$ and $\mathsf{Encaps}(ek_s)$ would produce two independent keys and the parties could never reconstruct the same $\rho=C\oplus G(K)$; delivering one shared chosen $K$ to both recipients is exactly encryption of a chosen message, i.e.\ a public-key encryption scheme. We therefore instantiate $\PKE$ with the CPA-secure module-LWE encryption underlying ML-KEM (Kyber), and not with the KEM itself. This encryption is precisely what the group construction obtains for nothing: there the shared key is the Diffie--Hellman value that both parties derive from their public keys, so \cite{liu2023} require no encryption at all, whereas we must transport it --- and to both parties --- which is why a public-key encryption scheme appears in our construction where theirs has none.

\paragraph{Consequences.}
The two schemes rest on parallel hardness assumptions: decisional MLWE replaces DDH/DBDH for withdrawability, and MSIS/SelfTargetMSIS replace CDH/DL for unforgeability. The lattice route is less compact --- an extra encryption key pair per party, and a withdrawable signature carrying two ciphertexts and a one-time pad rather than three group elements --- but it gains two properties: the withdrawable signature is pseudorandom, concealing from outsiders even the existence of a signature, and non-transferability against the designated verifier is perfect ($\varepsilon_{\mathrm{zk}}=0$ for the no-hint instantiation) rather than computational. We regard these as necessary adaptations of the Liu--Baek--Susilo notion to the post-quantum setting rather than departures from it.

\section{Conclusion and future research}
This work uses the ideas in \cite{liu2023} to extend the Fiat--Shamir with aborts paradigm \cite{lyubashevsky2009} with withdrawability and defines a general construction for withdrawable lattice-based digital signature schemes.

We demonstrated our approach by creating a withdrawable version of Dilithium, though the same principles could be applied to other signature schemes like HAETAE \cite{cheon2024}. Our construction maintains the security properties of the underlying signature scheme while adding the ability to withdraw signatures when needed by the signer.

Several directions remain for future research, including optimizing our construction's efficiency and exploring additional features such as blindness, multiparty capabilities, or enhancing this construction with extended withdrawability, where rather than limiting verification to a specific entity, we can ensure the universal verifiability of the withdrawable signature by employing any signature scheme that can maintain signer ambiguity (this is work done in \cite{liu2024}).

Another potential line of research is given by practical applications, particularly how these constructions could enhance quantum resistance in blockchain systems. The ability to withdraw signatures could prove valuable in blockchain environments where transaction revocation is desirable but traditionally difficult to implement.


\printbibliography

\end{document}